\documentclass[a4paper,10pt,reqno]{amsart}
\usepackage{amssymb, amsmath, graphicx, amsthm, enumitem, multicol,amsfonts}
\usepackage[margin=1in]{geometry}
\usepackage{cite}
\usepackage{ bbold }
\usepackage{longtable}
\usepackage{booktabs}
\usepackage{placeins}
\usepackage{hyperref}
\usepackage{tikz-cd}

\usepackage[utf8]{inputenc}

\DeclareMathOperator{\R}{\mathbb{R}}
\DeclareMathOperator{\C}{\mathbb{C}}

\DeclareMathOperator{\Z}{\mathbb{Z}}

\DeclareMathOperator{\F}{\mathbb{F}}

\DeclareMathOperator{\cH}{{\mathcal{H}}}

\newcommand{\Aut}{\mathrm{Aut}}

\newcommand{\im}{\mathrm{im}}
\newcommand{\coim}{\mathrm{coim}}

\newcommand{\ba}{\begin{align}}
\newcommand{\ea}{\end{align}}
\newcommand{\bea}{\begin{eqnarray}}
\newcommand{\eea}{\end{eqnarray}}
\newcommand{\be}{\begin{equation}}
\newcommand{\ee}{\end{equation}}

\newcommand{\id}{\mathbb{1}}

\newcommand{\ie}{{\it i.e.~}}
\newcommand{\eg}{{\it e.g.~}}

\newcommand{\diag}{{\mathrm{diag}}}

\newcommand{\pa}{\mathbb{p}}
\newcommand{\Si}[1]{{\Sigma_{#1}}}
\newcommand{\lSi}[1]{{\hat\Sigma_{#1}}}

\newcommand{\vac}{|0\rangle}

\newcommand{\Tr}{{\textrm{Tr}}}

\newcommand{\co}[1]{{\underline #1}}
\newcommand{\ct}[1]{{\underline {\check #1}}}

\newcommand{\hm}{{\mathbb{h}}}
\newcommand{\gm}{{\mathbb{g}}}

\newcommand{\Hr}{{F}}

\newcommand{\qb}{{\bar{q}}}

\numberwithin{equation}{subsection}

\title{CFTs with Large Gap from Barnes-Wall Lattice Orbifolds}
 \author{Christoph A.~Keller}
\address{Christoph A. Keller,  Department of Mathematics, University of Arizona, Tucson, AZ 85721-0089, USA}
 \email{cakeller@arizona.edu}
\author{Ashley Winter Roberts}
\address{Ashley Winter Roberts,  Department of Mathematics, University of Arizona, Tucson, AZ 85721-0089, USA}
 \email{awroberts@arizona.edu}

\author{Jeremy Roberts}
\address{Jeremy Roberts,  Department of Mathematics, University of Arizona, Tucson, AZ 85721-0089, USA}
 \email{jjrobertlol@arizona.edu}

\theoremstyle{plain}
\newtheorem{thm}{Theorem}[section]
\newtheorem{lem}[thm]{Lemma}
\newtheorem{prop}[thm]{Proposition}

\newtheorem*{theorem-non}{Theorem}
\newtheorem{con}[thm]{Conjecture}

\theoremstyle{definition}
\newtheorem{defn}{Definition}[section]

\theoremstyle{remark}

\begin{document}

 \begin{abstract}
We investigate orbifolds of lattice conformal field theories with the goal of constructing theories with large gap.
We consider Barnes-Wall lattices, which are a family of lattices with no short vectors, 
and orbifold by an extraspecial 2-group of lattice automorphisms. To construct the orbifold CFT, we investigate the orbifold vertex operator algebra and its twisted modules. To obtain a holomorphic CFT, a certain anomaly 3-cocycle $\omega$ needs to vanish; based on evidence we provide, we conjecture that it indeed does. Granting this conjecture, we construct a holomorphic CFT of central charge 128 with gap 4.
  
 \end{abstract}

\maketitle

\section{Introduction}
\subsection{Overview}

In this article we consider orbifolds of lattice conformal field theories. Motivated by $AdS_3/CFT_2$ holography, sparseness conditions such as \cite{Hartman:2014oaa}, and the proposal of extremal CFTs \cite{MR2388095, Witten:2007kt}, our goal is to construct theories with large gaps. That is, we want to construct CFTs whose lightest non-vacuum primary fields have relatively large weight.

How to construct orbifold CFTs is of course well known to physicists. However, some of the steps involve important subtleties and need to be done carefully and mathematically rigorously. Let us therefore start out with a very brief summary of how CFTs and their orbifolds are described mathematically by the theory of vertex operator algebras (VOAs) \cite{LL}.

A VOA $V$ describes the symmetry algebra of a 2d CFT, that is the set of its holomorphic (or anti-holomorphic) fields. The irreducible modules (or representations) $W_i$ of $V$ are the primary fields of the CFT. A \emph{rational} VOA has a finite number of such modules.  A \emph{holomorphic} VOA only has a single module, namely $V$ itself. It is therefore by itself already a holomorphic (or chiral) CFT, that is a CFT whose fields are all holomorphic. For a more general rational VOA we construct a CFT by taking a modular invariant combination of holomorphic and anti-holomorphic modules. That is, given a modular invariant $N_{ij}$, we construct a full CFT whose space of states is given by
\be
\cH = \bigoplus_{h,\bar h} \cH_{(h,\bar h)} = \bigoplus_{i,j}N_{ij}W_i \otimes \bar W_j\ .
\ee
where $(h,\bar h)$ denotes the holomorphic and anti-holomorphic conformal weights.
In general, given a VOA $V$, there are multiple modular invariants leading to different CFTs. The ADE classification of the Virasoro minimal models \cite{Cappelli:1986hf,Cappelli:1987xt,Kato:1987td} is probably the best known example of this. For a mathematical description of this program of constructing full CFTs, especially in the context of orbifold theories, see \cite{MR4256405}.

To construct orbifolds of $V$, we pick an automorphism group $G$ of $V$ and consider the sub-VOA  $V^G$, that is the set of states in $V$ invariant under $G$. Its irreducible modules can be obtained from the twisted sectors, that is the twisted modules of $V$. We can then again construct a CFT from modular invariant combinations of the modules of $V^G$. One construction that always works is to choose the diagonal modular invariant. This leads to a CFT whose states are invariant when $G$ acts simultaneously on the left- and right-movers. To put it another way, this is  a left-right-symmetric orbifold. However, even if $V$ is holomorphic, this orbifold is no longer holomorphic. To construct a holomorphic orbifold CFT, it is necessary to find another modular invariant.

With this general summary out of the way, let us define the \emph{gap} of a CFT. First take the case where $V$ is a holomorphic VOA. The gap of $V$ is then the smallest conformal weight $\mu_V$ of any non-vacuum Virasoro highest weight vector in $V$\cite{Hohn:2019luz}. Modularity of the partition function gives the bound $\mu_V\leq \lfloor \frac c{24}\rfloor +1$ \cite{MR1614941}. An \emph{extremal} VOA (or CFT) is a holomorphic VOA that saturates that inequality \cite{MR1614941, MR2388095, Witten:2007kt}. The most famous example of an extremal VOA is the Monster VOA $V^\natural$ \cite{MR843307,MR996026}. It has central charge $c=24$ and $\mu_V=2$. Other extremal VOAs are known up to $c=40$, but no examples for higher central charge are known. In particular, no holomorphic VOA with $\mu_V > 2$ is known.

Next consider non-holomorphic VOAs and their CFTs $\cH$. 
Here the gap is the smallest total conformal weight $h+\bar h$ of any non-vacuum vector in $\cH$ that is highest weight with respect to both the holomorphic and the anti-holomorphic Virasoro algebra $Vir\otimes \overline{Vir}$. In the special case when $V$ is holomorphic, $\cH= V\otimes \bar V$ and the gap is simply given by $\mu_V$ as above.
So far, no full CFT with gap greater than 2 is known. 

The goal of this article is to construct such CFTs.
To this end we construct orbifolds of Barnes-Wall lattice VOAs. Barnes-Wall lattices $BW(m)$ are a family of even lattices that are unimodular for $m$ odd \cite{MR0106893,MR335654,GriessRobertL2011Aitg}.
Because their shortest vectors are relatively long and their automorphism groups are big, they are suitable candidates for constructing CFTs with large gap.

In this article we obtain two main results: By orbifolding the Barnes-Wall lattice in $d=128$ by the extraspecial group $E(7)$, we first construct a CFT of central charge $c=128$ and gap 2 that only has a single highest weight vector of weight 2. To our knowledge this is the CFT with the fewest such states known in the literature. Second, we conjecture that we can use another modular invariant of this orbifold to construct a holomorphic CFT of central charge $c=128$ with gap 4. This construction only works if a certain anomaly 3-cocycle $\omega$ in $H^3(G,U(1))$ is trivial. We check various necessary conditions for the vanishing of $\omega$, providing evidence for the conjecture. 

Unfortunately this orbifold construction is rather involved. We therefore provide a less technical summary and outlook in the conclusion section at the end of the article. Readers who are not interested in the details of the construction of the CFT, but rather in its properties and possible applications, may want to skim through the rest of the introduction and then skip straight to section~\ref{s:conclusion}.

\subsection{The Monster VOA}
Our construction is somewhat lengthy and technical, but fundamentally proceeds along the lines of the construction of the Monster VOA $V^\natural$ in  \cite{MR996026}, that is the extremal VOA with $c=24$. As a roadmap let us therefore sketch the steps in their construction and compare them to our steps.

Their starting point is the VOA $V_L$ of the Leech lattice \cite{MR209983}, the even unimodular extremal lattice in $d=24$. 
The Leech lattice VOA has character
\be
\Tr_{V_L}q^{L_0} = 1+ 24 q + O(q^2)\ .
\ee
The 24 states of weight 1 all come from the free boson modes $a_{-1}^i\vac$. The Leech lattice is extremal, which means it has no vectors of length squared 2. The lattice states thus only contribute at order $q^2$ and higher. Next, to eliminate the states of weight 1, they orbifold by the lattice automorphism $z$ that maps $\vec x \mapsto -\vec x$. The character of the invariant sub-VOA is thus
\be
\Tr_{V_L^{\Z_2}} q^{L_0} = 1 + O(q^2)\ .
\ee
However, $V_L^{\Z_2}$ is no longer a holomorphic VOA. To turn it into one, they adjoin a module coming from the $z$-twisted module to construct $V^{orb(\Z_2)} = V^\natural$. The twisted module has high enough vacuum anomaly that it does not introduce any new states of weight 1. The character of $V^\natural$ is then famously
\be
\Tr_{V_L^{\natural}} q^{L_0} = 1 + 196884 q^2 + O(q^3)\ .
\ee
Thus $V^\natural$ indeed has gap 2 and is extremal.

We want to mimic this construction to obtain CFTs with gaps larger than 2. In lieu of the Leech lattice we use
Barnes-Wall lattices $BW(m)$,  a family of lattices of dimension $2^m$. For odd $m\geq3$, they are even and unimodular. For our purposes, they have two important properties: First, their shortest vectors have relatively big length squared, namely $2^{\lfloor m/2\rfloor}$. For $m=3$ we recover the $E_8$ lattice. For $m=5$, $BW(5)$ is extremal, which for $d=32$ simply means that there are no vectors of length squared 2. Some of the orbifolds in those two cases were studied in \cite{Gemunden:2019dtr,Burbano:2021loy} and \cite{Moller:2024plb} respectively. Even though for larger $m$ the Barnes-Wall lattices are no longer extremal, they still mimic the Leech lattice's property of not having very short vectors. For our construction we specifically use the $d=128=2^7$ BW lattice. Its lattice VOA has character
\be
\Tr_{V_L}q^{L_0} = 1+128 q+8384 q^2+374272 q^3+O(q^4) \ .
\ee
Next we use the fact that the Barnes-Wall lattices have large automorphism groups $BRW(m)$. The automorphism group $BRW(m)$ contains as a subgroup an extraspecial 2-group $2_+^{2m+1}$, which we will denote by $E(m)$; the full automorphism group is in fact an index 2 subgroup of the normalizer of $E(m)$ in the orthogonal group. In lieu of $\Z_2$, we orbifold by this larger, non-Abelian group $E(7)$. 

The reason for this is the following. 
To achieve a large gap CFT, we ask two things of an orbifold group: First, we want it to project out all the light non-Virasoro descendant states of $V$. As we will see, $E(m)$ is powerful enough to do this up to weight 4. For $m=7$ we have
\be\label{charVE7}
\Tr_{V_L^{E(7)}}q^{L_0} = 1+ q^2+ q^3+O(q^4) \ .
\ee
The states appearing at weight 2 and 3 are the Virasoro descendants $L_{-2}\vac$ and $L_{-3}\vac$. The smallest weight of a Virasoro highest weight state is thus indeed 4.

Finally we would like to add back modules to $V_L^{E(7)}$ to turn it back into a holomorphic VOA. To obtain such irreducible modules of a fixed point VOA $V^{G}$, we can first construct the twisted modules of $V$. A $g$-twisted module $W_g$ then carries a projective representation $\phi_g(\cdot)$ of the centralizer $C_G(g)$ of the orbifold group $G$ \cite{MR1794264, MR2023933,MR3715704}. That is, $\phi_g(\cdot)$ satisfies
\be
\phi_g(h_1)\phi_g(h_2)=c_g(h_1,h_2)\phi_g(h_1h_2)\ ,
\ee
where the multiplier $c_g(\cdot, \cdot)$ is given by a 2-cocycle in the Schur multiplier $H^2(G,U(1))$. All irreducible modules of $V^G$ can then be obtained from the twisted modules of $V$ by projecting to irreducible representations of $\phi_g(\cdot)$\cite{MR3715704, MR2040864}. 

Last, we want to make sure that the modules we adjoin do not reintroduce any states of weight less than 4. To ensure this, we require that all twisted sectors have a vacuum anomaly of at least 4. This is the second thing we ask of our orbifold group. As we will see in section~\ref{s:BWEm}, for $m\geq 7$ $E(m)$ has this property. This means that if we can add back twisted modules to obtain a holomorphic VOA, that VOA will have gap 4.
However, as mentioned above, this is possible only if a certain anomaly 3-cocycle $\omega$ vanishes. Let us discuss this and other issues next.

\subsection{The steps in more detail}
In the previous section we gave an outline of the steps of our construction. As mentioned in the beginning however, several of these steps are mathematically quite subtle and need to be done carefully.
Let us therefore go over the procedure outlined above more carefully. In particular we want to point out three  issues that we need to deal with.

The first issue is that to orbifold $V_L$, we need to work with automorphisms of the VOA $V_L$, and not of the lattice $L$. To find such automorphisms, we need to lift automorphisms of $L$ to automorphisms of $V_L$, or equivalently to automorphisms of the twisted group algebra of $L$ \cite{MR1745258}. For a given  automorphism $g$ of $L$ we can always do this to obtain an automorphism $\hat g$ of $V_L$. However, there is no guarantee that the $\hat g$ will satisfy the same group relations as the $g$. A typical example of this is order doubling, in which case the order of 
$\hat g$ is twice the order of $g$ \cite{vanEkeren2017, Harvey:2017rko}.
What this means is that the $\hat g$ will no longer generate the group $G$, but rather a cover of it. In principle such a cover is a perfectly good group of automorphisms of $V_L$ that could be used for orbifolding. For our purposes however having a cover is bad: if there are multiple preimages of the identity, then those will give light twisted sector vacuum anomalies, which we want to avoid. We therefore need to make sure that our extraspecial group $E(m)$ lifts to $E(m)$; that is, we need to ensure that the lifted automorphisms still satisfy the relations of the extraspecial group. Theorem~\ref{thm:Elift} in section~\ref{s:BWEm} shows that this is the case, so that $E(m)$ is indeed a group of automorphisms of $V_L$. 

This brings us to the second issue.
The result above allows us to work with $E(m)$ as a group of automorphisms of $V_L$. In particular we can construct $V_L^{E(m)}$ and its modules from the twisted modules $W_g$ of our lattice VOA.
For this we use the fact that for lattice VOAs, the twisted modules $W_g$ can be constructed as induced representations
of a smaller representation $\rho$ on $\Omega_0$, the so called defect representation \cite{MR820716, MR2172171}. $\rho$ is a projective representation of a finite group $N$ that is obtained as a quotient of lattice groups.
To then obtain the actual irreducible modules and their characters, we need to construct the projective representation $\phi_g(\cdot)$ of $C_G(g)$ on $W_g$. Even though we know for general reasons that the $\phi_g(\cdot)$ exist abstractly, computing them concretely takes some work.
We do this by first restricting the action of $\phi_g(\cdot)$ to $\Omega_0$. The restricted representation we then obtain by constructing intertwiner maps $T_h$: We use the fact that for an automorphism $h\in C_G(g)$, $\rho$ and $\rho \circ h$ are projectively equivalent, that is related by an isomorphism $T_h$. These $T_h$ then form the sought after restriction of the projective representation of $\phi_g(\cdot)$ to $\Omega_0$. 

The main computational work here is to find the intertwiner maps $T_h$. In section~\ref{s:defrep} we construct these by using Schur averaging, that is by averaging over all elements of $N$ to obtain an $N$-invariant map between $\rho$ and $\rho \circ h$. However, the order of $N$ grows very quickly in $m$; it tends to be something like $2^{2^m}$. Evaluating the sum even for $m=5$ or 7 becomes very cumbersome. We therefore describe a method that reduces the summation to a linear algebra problem in a $2^m$ dimensional vector space over $\F_2$. This problem can then be solved in reasonable time by standard computer algebra systems such as Magma or Mathematica. This allows us to construct the intertwining representation $T_h$.

In principle we could then proceed to lift the representation $T_h$ on $\Omega_0$ to the representation $\phi_g(\cdot)$ on $W_g$. In particular this would allow us to compute the characters of all the irreducible modules of the orbifold. We did not attempt this. 
Instead, for reasons explained below, we focus on the multiplier $c_g(\cdot,\cdot)$ of the representation, which we can obtain by comparing a single entry of the products of intertwiner matrices.

Having constructed the twisted modules, we first use the diagonal modular invariant to construct CFTs. For $m=7$, this is the content of proposition~\ref{prop:diag}. Unfortunately the resulting CFT still has only gap 2. The reason for this is that it contains the  state $a_{-1}\cdot \bar a_{-1}\vac$ which is invariant under the diagonal action of the orbifold group. However, this is the only highest weight state of this weight, meaning there is only one primary of total weight 2.

To get around this problem we investigate holomorphic extensions of $V_L^{E(7)}$. 
As argued above, if such an extension exists, up to order 4 its character is the same as (\ref{charVE7}). This would mean that we have indeed constructed a holomorphic CFT with gap 4.  
This however brings us to the third issue, namely if we can find a holomorphic extension of $V_L^{E(7)}$. 
Because $V_L^{E(7)}$ is strongly rational, its modules form a modular tensor category \cite{MR2387861,MR2468370}. This modular tensor category is given by a twisted Drinfeld double of the group $G$, $D^\omega(G)$-mod \cite{dijkgraaf1990,Roche:1990hs}. The twist $\omega$ is a $3$-cocycle $\omega \in H^3(G;U(1))$. 
By the work of \cite{Evans:2018qgz}, finding a holomorphic extension $V^{orb(G)}$ is only possible $\omega$ if is trivial. In physics $\omega$ is sometimes called the 't Hooft anomaly. Its relevance for holomorphic orbifolds was discussed in \cite{FV87}. We conjecture that for the orbifold of the $d=128$ BW lattice VOA by $E(7)$, $\omega$ is trivial. We will not prove this here, but we do provide two important pieces of evidence in favor of this conjecture.

The first piece of evidence comes from restricting $\omega$ to cyclic subgroups of $G$. There is a simple criterion for when the restriction of $\omega$ to a cyclic subgroup is trivial: \cite{vanEkeren2017} call this the `type 0 condition', and the physics literature calls it `level matching' \cite{Vafa:1986wx}. We check that this condition is satisfied for our $E(7)$ orbifold.

The second piece of evidence comes from considering the multipliers $c_g(\cdot,\cdot)$ of the twisted module representations $\phi_g(\cdot)$. These 2-cocycles are related to $\omega$ via \cite{Evans:2018qgz}
\be
c_g(h_1,h_2) = \omega(g,h_1,h_2)\omega(h_1,h_2, h_2^{-1}h_1^{-1}gh_1h_2)\omega(h_1,h_1^{-1}gh_1,h_2)^*
\ee
In particular, if $\omega$ is trivial, necessarily all multipliers $c_g(\cdot,\cdot)$ also need to be trivial. That is, the projective representations of the centralizers $C_G(g)$ on the $g$-twisted modules are actually (after an appropriate choice of section) linear representations. We check that for our $E(7)$ orbifold, all $c_g(\cdot,\cdot)$ are trivial, providing our second piece of evidence in favor of the conjecture.

Granting that $\omega$ is trivial, there is indeed a holomorphic extension of $V^{E(7)}$. This in turn implies that we can construct a holomorphic CFT of central charge $c=128$ that has gap 4.
To our knowledge this is the first example of a CFT in two dimensions with gap larger than 2. As such, it is of interest in physics for instance in the context of $AdS_3/CFT_2$ holography, where having a sparse light spectrum is desirable \cite{Hartman:2014oaa}. Even though our CFT is not extremal, it has fewer light states than theories previously constructed from lattice orbifolds \cite{Gemunden:2018mkh,Gemunden:2019dtr} or permutation orbifolds \cite{Belin:2014fna,Haehl:2014yla,Belin:2015hwa,Keller:2017rtk}. Outside of holography, it can be used to test conjectures such as \cite{Collins:2022nux}. 

This article is structured in the following way. In section~2 we discuss some background on the construction of VOAs from lattices and orbifolds. In section~3 we introduce the Barnes-Wall lattices, their automorphism groups, the extraspecial 2-group $E(m)$ and its properties, and in particular its lift to automorphisms of the VOA. In section~4 we introduce the intertwining representation. We explain how to compute the intertwiner matrices explicitly and check that the 2-cocycle vanishes. In section~5 we then return to the construction of orbifold VOAs and apply these results to construct CFTs with large gaps. Finally in section~6 we give a less technical summary of the constructed CFT and its properties.

\section{Orbifolds of lattice VOAs}

\subsection{Orbifold theory}\label{ss:oftheory}

Let us start with a summary of the mathematical theory of VOA orbifolds.
Let $V$ be a strongly rational VOA, that is rational, $C_2$-cofinite, simple, self-contragredient and of CFT-type. The significance of this is that \cite{MR2387861} established that the fusion rules for the modules of a strongly rational VOA satisfy the Verlinde formula \cite{Verlinde:1988sn} and hence that the modules
form a modular tensor category \cite{MR2468370}.
Next let $G$ be a finite subgroup  of $\Aut(V)$. Then $V^G$, the sub-VOA invariant under $G$, is again strongly rational as long as $G$ is solvable
by the combined results of \cite{MR1684904, MR3320313, 2016arXiv160305645C}. In our case $G$ will be an extraspecial group and therefore solvable. 

If we moreover assume that $V$ is holomorphic, then for a given $g \in G$ there is a unique $g$-twisted module $W_g$. 
The centralizer $C_G(g)$ of $g$ has a projective representation $\phi_g(\cdot)$ on $W_g$ with some cocycle $c_g(\cdot,\cdot)$ \cite{MR1794264, MR2023933,MR3715704}.
From \cite{MR1794264} we know that $W_g$ decomposes into modules  for $\C_{c_g}[C_G(g)]\otimes V^G$ as $C_G(g)$ through a Schur-Weyl duality
\begin{equation}\label{eq:SW_1}
W_g = \bigoplus_{\chi \in \textrm{Irr}_{c_g}(C_G(g))} V_{[g,\chi]} \otimes W_{[g,\chi]},
\end{equation}
where $\textrm{Irr}_{c_g}$ denotes the set of irreducible projective characters with $2$-cocycle $c_g$, $V_{[g,\chi]}$ denotes the irreducible projective representation of the centraliser $C_G(g)$ with character $\chi$ and $W_{[g,\chi]}$ is the corresponding irreducible module for the fixed-point subalgebra $V^G$.
\cite{MR3715704, MR2040864} then establishes that all irreducible modules of $V^G$ are given by $W_{[g,\chi]}$ for some conjugacy class $g$ and some projective irreducible representation $\chi$ of $C_G(g)$. In summary, the irreducible modules are labeled by $[g,\chi]$ and can be obtained from the twisted modules using the projector
\begin{equation}\label{eq:SW_2}
\pi_{\chi} = \frac{1}{|C_G(g)|}\sum_{h\in C_G(g)} \overline\chi(h) \phi_g(h).
\end{equation}
From this we see that although $V$ is holomorphic, $V^G$ is no longer holomorphic. Its irreducible modules given above form a modular tensor category governed by a $3$-cocycle $\omega \in H^3(G;U(1))$, isomorphic to a twisted Drinfeld double of the group $G$, $D^\omega(G)$-mod. In physics $\omega$ is sometimes called the 't Hooft anomaly. This was originally proposed in \cite{dijkgraaf1990,Roche:1990hs}.
It  was then proven for a special case in  \cite{MR1923177} and for the general case in \cite{Dong:2021yrv}. (See also \cite{Gannon:2024tcl}.)

The cocycle $\omega$ determines the multiplier of the projective representations of the twisted modules:
the cocycles $c_g, g \in G$ of the $C_G(g)$ representations $\phi_g(\cdot)$ are given by
\be\label{descend}
c_g(h_1,h_2) = \omega(g,h_1,h_2)\omega(h_1,h_2, h_2^{-1}h_1^{-1}gh_1h_2)\omega(h_1,h_1^{-1}gh_1,h_2)^*\ .
\ee
Similarly $\omega$ determines the $S$ and $T$ transformation matrices \cite{Evans:2018qgz}.

The next step is then to construct the full CFT from the characters of the irreps of $V^G$. 
See \cite{MR4256405} for a review on constructing orbifold conformal field theories from representations of VOAs. To construct a full conformal field theory from a VOA, we first need to construct a chiral conformal field theory by choosing a category of modules and intertwining operators.\footnote{Note that this terminology differs from physics, where a chiral theory is often taken to mean a holomorphic CFT, that is a holomorphic VOA.} This chiral conformal field theory needs to be combined with an antichiral copy to give the full conformal field theory.
To put it another way, we want to find modular invariants to construct the full CFT. By the strong rationality of $V^G$, we can always use the diagonal modular invariant for this \cite{MR2300247,MR2592945}
\be\label{Hdiag}
\cH = \bigoplus_{[g,\chi]} W_{[g,\chi]}\otimes \bar W_{[g,\chi]}
= \bigoplus_{[g]} \left( W_g \otimes \bar W_g \right)^{(C_G(g)\times C_G(g))^{diag}}\ .
\ee
To see this identity, note that it can be written as the projection of all twisted sectors to states invariant under the diagonal group $(C_G(g)\times C_G(g))^{diag}$. This follows from (\ref{eq:SW_1}) and the fact that the $V_{[g,\chi]}$ are irreducible. In physics this  CFT is often called the (left-right) symmetric orbifold of $V$.

Under the right circumstances we can also use other modular invariants to construct the full CFT. These correspond to extensions of $V^G$ to larger VOAs. In particular one can hope to find a modular invariant that allows one to extend the orbifold VOA $V^G$ back to a holomorphic VOA $V^{orb(G)}$. Whether such an extension is possible depends on a 3-cocycle $\omega \in H^3(G,U(1))$. 
Only if it is trivial can we find a holomorphic extension and hence construct a holomorphic VOA again. Its modular invariant is then
\be\label{VorbG}
\cH = V^{orb(G)}\otimes \bar V^{orb(G)} =  \bigoplus_{[g]} W_g^{C_G(g)} \otimes  \bigoplus_{[g]} \bar W_g^{C_G(g)}\ .
\ee
To prove that $\omega$ is trivial we would have to analyze the modular data of the MTC. Instead we will only test two necessary conditions for $\omega$ to be trivial. First, we check that it is trivial when restricted to any cyclic subgroup $\langle g\rangle$ of $G$. By the results of \cite{vanEkeren2017} this can be done by simply computing the vacuum anomaly of the $g$-twisted module.
Second we use the relation (\ref{descend}) between $\omega$ and the projective representations of $C_G(g)$ on the twisted modules. 
Since the $c_g$ are the 2-cocycles of the projective representation of $C_G(g)$ for the $g$-twisted module, this implies that if $\omega$ is to be trivial, all projective representations must have trivial cocycles.

\subsection{Lattice VOAs and twisted group algebras}\label{ss:TGA}
Let us apply the above to the concrete example of lattice VOAs. Given a positive definite even lattice $L$, we can construct an associated strongly rational VOA \cite{MR843307,MR996026}. The main ingredient for this is the twisted group algebra.
The twisted group algebra $\C_\epsilon[L]=:\hat L$ has multiplication $e_{\alpha}e_{\beta}=\epsilon(\alpha,\beta)e_{\alpha+\beta}$ where the 2-cocycle $\epsilon$ evaluated on the basis vectors is $\epsilon(\alpha_i,\alpha_j)= (-1)^{H_{ij}}$, where $H$ is the `half-Gram matrix' defined as
\be
H_{ij} = \left\{ \begin{array}{cc} G_{ij} &: i>j\\ G_{ij}/2 &: i=j \\ 0 &:i<j \end{array} \right . 
\ee
Next we want to construct automorphisms of $\C_\epsilon[L]$ by lifting automorphisms of the lattice; these will be automorphisms of the VOA  \cite{MR1745258}.
Given an automorphism $g\in Aut(L)$, define the bilinear form $B_g$ as
\be\label{Bg}
B_g(\alpha_i,\alpha_j) = \left\{ \begin{array}{cc} H(\alpha_i,\alpha_j) - H(\alpha_ig,\alpha_jg) &: i>j\\ 0 &:i\leq j \end{array} \right. ,
\ee
and the function
\be
u_g(\alpha) = (-1)^{B_g(\alpha,\alpha)}\ .
\ee
Note that we take the convention that $g$ acts from the right.
The map $\hat g$ acting on vectors $e_\alpha$ in $\C_\epsilon[L]$ as
\be
\hat g (e_\alpha) = u_g(\alpha)e_{\alpha g}
\ee
is then an automorphism of the twisted lattice algebra, $\hat g \in Aut(\hat L)$.
Finally, $\hat g$ is called a standard lift if $u_g=1$ on the fixed point lattice $L_g$\cite{MR820716}.

\subsection{Twisted modules and the defect representation}\label{ss:dr}
We now want to construct the $\hat g$-twisted modules of the lattice VOA $V_L$.
For standard lifts $\hat g$, the twisted modules $W_{\hat g}$ for lattice VOAs were constructed in \cite{MR820716, MR2172171} as induced representations
of a smaller representation $\Omega_0$ of a finite group, the so called defect representation. We will briefly
review this construction, using the notation of \cite{ MR2172171}. We will then also explain the idea behind constructing the projective representation $\phi_{\hat g}$. Twisted modules of non-standard lifts can be constructed similarly using \cite{MR1372724}. As we will see in section~\ref{s:orbifolds} however, all our lifts turn out to be standard lifts, so that we will only describe the standard case in what follows.

 Let $g\in Aut(L)$ and $\hat g$ its lift to $Aut(V_L)$. Define the group 
\be
N:= L^\perp_g/L(1-g)
\ee
where $L^\perp_g$ is the lattice of points in $L$ orthogonal to the fixed point lattice $L_g$. $N$ is a finite Abelian group. To construct it explicitly, define $D:=1-g$ and find its Smith Normal Form
\be
S = PD Q \ , \qquad S = \diag(s_1,s_2,\ldots, s_d,0,\ldots, 0)\ ,
\ee
with $s_i|s_{i+1}$ and $s_i \in \Z_+$. If $\{\alpha_1, \ldots, \alpha_d,\delta_1,\ldots ,\delta_{\dim L - d}\}$ are the row vectors of $Q^{-1}$, then $N$ is generated  by the $[\alpha_i]$ and has the form
\be
N \cong \bigoplus_{i=1}^k (\Z / s_i\Z)\ .
\ee
We can construct a nondegenerate alternating bilinear form on $N$ as
\be
C(\alpha,\beta) =  \prod_{k=0}^{n-1} (-\xi_n^k)^{-\langle\alpha g^k|\beta\rangle}\ 
\ee
where $n$ is the order of $g$ and $\xi_n$ is a primitive $n$-th root of unity. 

We now construct a projective representation $(\rho,\Omega_0)$ of $N$ whose 2-cocycle has skew $C$ following \cite{MR2742735}. The idea of the construction is as follows: We pick generators $\alpha_i, \beta_i$ of $N$ which are orthogonal to each other except for $C(\alpha_i,\beta_i)$. We then interpret the $\alpha_i$ as generators of some subgroup $A<N$, and the $\beta_i$ as generators of $\check A$, the group of irreducible characters of $A$, giving an isomorphism
\be
\gamma: N \to A+\check A \qquad [\mu] \mapsto (x_{[\mu]},\chi_{[\mu]})\ .
\ee
$\Omega_0$ is then spanned by vectors $e_\psi, \psi \in \check A$, and $\rho$ is the action
\be
\rho(x,\chi) e_\psi = \psi(x)e_{\chi\psi}\ .
\ee
Its 2-cocycle is given by $P(\alpha,\beta)=\psi_\beta(x_\alpha)$, and its skew is given by $C(\alpha,\beta)=P(\alpha,\beta)P(\beta,\alpha)^{-1}$.  Moreover this representation is irreducible, and the unique irreducible representation with this 2-cocycle \cite{MR3299063}.

We can now consider the action of an element $\hat h$ in the centralizer of $\hat g$ on $\rho$. We have
\be
\rho_{\hat h}(\alpha) := u_h(\alpha) \rho \circ h (\alpha)
=u_h(\alpha) \rho(\alpha h)\ .
\ee
We then have that $\rho_{\hat h}$ is again a projective representation of  $N$ with 2-cocycle $P_{\hat h}(\alpha,\beta)$ and the same skew and the same dimension. Because their skew is the same, we have that $P_{\hat h}(\cdot,\cdot)/P(\cdot, \cdot)$ is symmetric. Moreover for the automorphisms $\hat h \in H$ that we consider we find that its diagonal entries vanish. It follows that $P$ and $P_{\hat h}$ are cohomologous --- see section~\ref{ss:coordinate} for an explicit construction of the coboundary. (In fact, on general grounds we expect this from the existence of the projective representation $\phi_{\hat g}$.) In section~\ref{s:defrep} we will explain how to use this to construct an intertwining representation $T$ of $H$.  
$\phi_{\hat g}$ can then obtained by extending $T$ from $\Omega_0$ to all of $W_{\hat g}$ \cite{Gemunden:2019dtr,Gemunden:2020lay}. In particular, the 2-cocycle of $\phi_{\hat g}$ is cohomologous to the 2-cocycle of $T$.
If we can therefore show that the cocycle $c(\hat h_1,\hat h_2)$ of $T$ is trivial, it will follow that $\phi_{\hat g}$ also has trivial 2-cocycle.

\section{Barnes-Wall lattices and the extraspecial group}\label{s:BWEm}
\subsection{Barnes-Wall lattices and their automorphism groups}
Let us introduce the main players of our construction, the Barnes-Wall lattices $BW(m)$, their automorphism groups $BRW(m)$ and the extraspecial 2-groups $E(m) = 2_+^{2m+1}$. The Barnes-Wall lattices were constructed in \cite{MR0106893} and independently in \cite{MR335654}; see \cite{GriessRobertL2011Aitg} for an overview.

The Barnes-Wall lattices $BW(m)$ are a family of integral lattices of dimension $2^m$. For $m>1$, $BW(m)$ is even, and hence suitable for constructing a VOA. For odd $m\geq 3$, it is unimodular, and for $m$ even its Gram matrix has determinant $2^{2^{m-1}}$. This means for odd $m\geq 3$ we can use $BW(m)$ to construct holomorphic lattice VOAs.
The other important property of Barnes-Wall lattices is that their shortest vectors have length squared $2^{\lfloor m/2\rfloor}$, which makes them suitable for constructing VOAs with large gap.

For our computations we will use the simple construction of $BW(m)$ given in \cite{NR02}. The basic idea is to introduce so-called balanced BW lattices whose generator matrices are given recursively by tensor products. $BW(m)$ is then obtained as the integer sublattice of the balanced lattice. We give detailed expressions for this construction in appendix~\ref{s:BW}.

Next let us turn to the automorphism group $BRW(m)=Aut(BW(m))$ first studied in \cite{MR125874,MR142666,MR148736}. We will first introduce a subgroup of it, namely the extraspecial 2-group $E(m) \simeq 2_+^{2m+1}$. We construct it as a subgroup of the orthogonal group $O(2^m,\R)$ in the following manner. Let 
\be
\sigma_1 = \begin{pmatrix}0 & 1\\1&0 \end{pmatrix}\ , \qquad \sigma_2 = \begin{pmatrix}1 & 0\\0&-1 \end{pmatrix}\ .
\ee
Then $E(m)$ is generated by $z=-\id$ and the $2m$ generators
\be
\Si{ar}= \id_2 \otimes \cdots \id_2 \otimes \sigma_r \otimes\id_2 \cdots \otimes \id_2 \qquad r=1,2\qquad a =1,2,\ldots, m\ .
\ee
It is straightforward to check that they indeed satisfy the defining relations of $2_+^{2m+1}$,
\be
[ z, \Si{ar}]=1 \qquad  z^2 =  \Si{ar}^2=1 \qquad [ \Si{ar}, \Si{bs}]=  z^{\delta_{a,b}\delta_{r,s+1\mod 2}}\ ,
\ee
where $[x,y]=xyx^{-1}y^{-1}$. The elements in $E(m)$ are then automorphisms of $BW(m)$. 
For the full automorphism group $BRW(m)$ we will follow \cite{MR1845897}. It can be obtained from the normalizer of $E(m)$ in $O(2^m,\R)$. For $m>3$, its structure is 
\be
BRW(m) = E(m).\Omega^+(2m,2)\ .
\ee
More precisely, the real Clifford group $C_m$ is defined as the normalizer of $E(m)$ in $O(2^m,\R)$. To describe its structure, we first note that $E(m)/Z(E(m))\simeq \F_2^{2m}$ is a vector space over $\F_2$. This vector space has a quadratic form $Q$ defined as $Q(x)=x^2 \in \F_2$, where the square of $x\in E(m)$, which is $\pm\id$, is interpreted as $\pm1 \in \F_2$. The automorphism group preserving $Q$ is $O^+(2m,2)$. The Clifford group is then given by $C_m=E(m).O^+(2m,2)$. By corollary 2.3 of \cite{MR1845897}, it is generated by the following elements of  $O(2^m,\R)$:
\be
\sigma_1\otimes \id_2 \otimes \ldots\otimes \id_2 , \ \sigma_2\otimes \id_2 \otimes \ldots\otimes \id_2 , \
H\ ,
GL(m,2) ,\ D(m)\ .
\ee
Here $GL(m,2)$ acts by permuting basis vectors of $\R^{2^m}$. That is, the basis vectors of $\R^{2^m}$ are labeled by elements of $F_2^m$ on which elements of $GL(m,2)$ act as bijections, hence giving a permutation of the basis vectors. 
Moreover we defined
\be
H = h \otimes\id_2 \otimes \ldots\otimes \id_2\ , 
\ee
with $h=\frac1{\sqrt 2} \begin{pmatrix}1&1\\1&-1\end{pmatrix}$
and
\be
D(m) = \textrm{diag}(1,1,1,-1)\otimes \id_2 \otimes \ldots \otimes \id_2\ .
\ee
The actual automorphism group $BRW(m)$ is an index 2 subgroup of $C_m$. It is given by the kernel $\Omega^+(2m,2)$ of the Dickson invariant $D: O^+(2m,2)\to \F_2$, so that indeed $BRW(m)=E(m).\Omega^+(2m,2)$. For a block matrix $U=\begin{pmatrix}A&B\\C&D\end{pmatrix}\in O^+(2m,2)$ the Dickson invariant is given by $D(U)=\Tr(B^TC)$ \cite{MR85212}. Using this we find all generators of $C_m$ have Dickson invariant 0 except for $H$. $BRW(m)$ is thus given by all products of generators of $C_m$ with an even total number of $H$.

In what follows we will focus on the extraspecial subgroup $E(m)<BRW(m)$. It turns out to be sufficient for the orbifold VOAs that we want to construct. The only time we will make use of the full group $BRW(m)$ is when computing conjugacy classes of $E(m)$ within $BRW(m)$ in appendix~\ref{ss:Emcon}. We therefore leave our discussion of $BRW(m)$ somewhat sketchy and refer to \cite{MR1845897,GriessRobertL2011Aitg, PhDJeremy} for details.

\subsection{Some properties of $E(m)$}\label{ss:Efacts}
Let us gather some facts about $E(m)$ and some of its representations. For the purposes of constructing twisted modules, we introduce the following well known definition, see for instance \cite{vanEkeren2017}.
\begin{defn}
Let $g\in O(V)$ of order $n$. The eigenvalues of $g$ are then given by $n$-th roots of unity, $e^{2\pi k/n}$ with $k=0,1,\ldots n-1$, with $V_{(k)}$ the corresponding eigenspaces. We define 
\be\label{vacanomaly}
\rho_g = \frac1{4n^2}\sum_{k=0}^{n-1}k(n-k)\dim V_{(k)}\ .
\ee

Next we need the following result on conjugacy classes and centralizers of $E(m)$ within $BRW(m)$:
\begin{prop}\label{prop:Emcc}
For $m\geq2$, under conjugation by elements of $BRW(m)$, $E(m)$ has 4 conjugacy classes with representatives $\id,-\id, \sigma_2\otimes\id_2\otimes\cdots \otimes \id_2$ and $\sigma_1\sigma_2\otimes\id_2\otimes\cdots\otimes\id_2$. Their eigenvalues $\lambda$ and other properties are given by
\begin{align}
|g| && \lambda && \dim V_{(k)} && \rho_g && |Cl_{BRW(m)}(g)|\nonumber \\
1 && 1 && 2^m && 0 && 1\\
2 && -1 && 2^m && 2^{m-4}&& 1\\
2 && 1,-1 && 2^{m-1} && 2^{m-5}&& 4^m+2^m-2 \label{S11conj}\\
4 &&  i,-i && 2^{m-1} && 3\cdot2^{m-6}&& 4^m-2^m
\end{align}
Their centralizers in $E(m)$ are isomorphic to $E(m),E(m), \Z_2\times E(m-1)$ and $\Z_4\circ E(m-1)$.
\end{prop}
\begin{proof}
The main part of this result is stated without proof in the proof of Theorem~14 in \cite{MR2269551}. We give a sketch of a proof here.
A straightforward computation shows that the four representatives have the order and eigenvalues given in the table. Because their spectrum differs, they clearly belong to distinct conjugacy classes. Moreover lemma~\ref{lem:mcon} in the appendix shows that any element in $E(m)$ can be brought to the form of the four elements above, establishing the claim. As for the centralizers, the first two are obvious. The centralizer of the third representative is a product of groups
\be
\langle \sigma_2\otimes\id_2\cdots \otimes \id_2\rangle\cdot\{\pm \id_2\otimes \sigma_1^{a_2}\sigma_2^{b_2}\otimes\cdots \otimes \sigma_1^{a_m}\sigma_2^{b_m}:a_i,b_i\in \{0,1\}\}\cong \Z_2\times E(m-1)
\ee
and for the fourth representative it is 
\be
\langle \sigma_1\sigma_2\otimes\id_2\cdots \otimes \id_2\rangle\cdot \{\pm \id_2\otimes \sigma_1^{a_2}\sigma_2^{b_2}\otimes\cdots \otimes \sigma_1^{a_m}\sigma_2^{b_m}:a_i,b_i\in \{0,1\}\}
\cong \Z_4 \circ E(m-1)\ ,
\ee
the central product coming from the fact that $-\id$ is an element of both groups in the product. For the counting, we note that there are $2\cdot \binom{m}{k} 3^{m-k}$ elements that have $k$ factors $\sigma_1\sigma_2$. Since an element of order 4 has to have an odd number of such factors, we have
\be
\sum_{k \textrm{\ odd}}^m \binom{m}{k}3^{m-k} = 2 \cdot \frac12((3+1)^m-(3-1)^m = 4^m -2^m\ .
\ee
The size of the conjugacy class of the non-central order 2 element can be obtained from a similar argument, or simply by observing that the total number of elements is $2^{2m+1}$.
\end{proof}

Next we introduce the representation of $E(m)$ that is spanned by the Heisenberg algebra of the lattice VOA.
\begin{defn}\label{defn:Hr}
Define the \emph{Heisenberg representation} (or Fock representation) $\Hr$ of $G<O(V)$ as the graded representation
\be
\Hr = \bigoplus_{n\geq 0} \Hr_{(n)}
\ee
whose graded character
\be
\chi^\Hr_g(q) := \sum_{n=0}^\infty \chi^{\Hr_{(n)}}_g q^n 
\ee
is given by
\be\label{FockCharacter}
\chi^\Hr_g(q) = \prod_{i=1}^{\dim V} \prod_{k=1}^\infty \frac1{1-\lambda_i^g q^k}
\ee
where $\lambda^g_i$ are the eigenvalues of $g \in O(V)$.
\end{defn}
We will be interested in invariant subspaces $(\Hr_{(i)}\otimes \Hr_{(j)})^{E(m)}$ for the first few $i,j$. To this end we compute how often the trivial representation appears in this tensor product. We therefore establish the following proposition:

\begin{prop}\label{prop:HeInvar}
For $m\geq 2$  we have 
\bea
\langle \chi^\Hr_{(1)},\chi^\Hr_{(0)}\rangle = 0 && 
\langle \chi^\Hr_{(2)},\chi^\Hr_{(0)}\rangle = \langle \chi^\Hr_{(3)},\chi^\Hr_{(0)}\rangle=1 \\
\langle \chi^{\Hr_{(1)}}, \chi^{\Hr_{(1)}}\rangle = 1\ , && \langle \chi^{\Hr_{(1)}}, \chi^{\Hr_{(2)}}\rangle=1 \ .
\eea
\end{prop}
\begin{proof}
As usual the inner product is given by
\be
\langle \chi^1,\chi^2\rangle = \frac1{2^{2m+1}}\sum_{g\in E(m)} \chi_g^{1} \chi_g^2\ .
\ee
A straightforward expansion of (\ref{FockCharacter}) gives the first three characters as
\bea
\chi_g^{\Hr_{(1)}} &=& \sum_i \lambda^g_i\\
\chi_g^{\Hr_{(2)}} &=& \sum_i\left( \lambda^g_i + \frac12(\lambda^g_i)^2\right) + \frac12\sum_{i,j} \lambda^g_i \lambda^g_j\\
\chi_g^{\Hr_{(3)}} &=& \sum_i\left( \lambda^g_i + \frac13(\lambda^g_i)^3\right) + \sum_{i,j} (\lambda^g_i \lambda^g_j+\frac12(\lambda^g_i)^2\lambda^g_j)
+\frac16\sum_{i,j,k} \lambda^g_i \lambda^g_j\lambda^g_k
\eea
By inspecting the table above we find for the elements of the four classes
\begin{eqnarray}
\sum_i \lambda^g_i =\sum_i (\lambda^g_i)^3&= 2^m ,\ -2^m ,\ 0 ,\ 0\\
\sum_i (\lambda^g_i)^2 & = 2^m, 2^m, 2^m, -2^m
\end{eqnarray}
Using this we find indeed
\be
\sum_g \chi_g^{\Hr_{(1)}}=0\ , \qquad
\sum_g \chi_g^{\Hr_{(2)}} = 2^{2m+1}\ , \qquad
\sum_g \chi_g^{\Hr_{(3)}} = 2^{2m+1}\ ,
\ee
and
\be
\sum_g \chi_g^{\Hr_{(1)}}\chi_g^{\Hr_{(1)}}= 2^{2m+1}\ , \qquad \sum_g \chi_g^{\Hr_{(1)}}\chi_g^{\Hr_{(2)}} =2^{2m+1}
\ee
establishing the claim.
\end{proof}

\end{defn}

\subsection{Lifting $E(m)$}\label{ss:Elift}
So far we have only discussed automorphisms $g$ of the underlying Barnes-Wall lattices. To obtain automorphisms of the lattice VOA $V_L$, we first lift $g$ to an automorphism $\hat g$ of the twisted group algebra $\hat L$. Section~\ref{ss:TGA} gives an explicit expression of how to do this by adjoining a function $u_g$. The lifted automorphism $\hat g$ is then also an automorphism of $V_L$; in fact, $Aut(V_L) = N \cdot O(\hat L)$ where $O(\hat L)$ is the group of all such lifted lattice automorphism and $N$ is the inner automorphism group \cite{MR1745258}. (In fact in our case $N< O(\hat L)$ since our $L$ does not have vectors of length squared 2.)

Let us discuss this lifting in more detail. Pick a subgroup $G<Aut(L)$. By the above, every $g\in G$ can be lifted to $\hat g \in Aut(V_L)$. We can use these to generate a group
\be
\hat G := \langle \hat g : g\in G \rangle < Aut(V_L)\ .
\ee
In general, $\hat G$ will not be isomorphic to $G$, but rather a cover of $G$. This happens if the lifted $\hat g$ no longer satisfy the same relations as the $g$. A typical example of this is order doubling, in which case the order of $\hat g$ is twice the order of $g$ \cite{vanEkeren2017}.

In principle there is nothing wrong with $\hat G$ being larger than $G$: we could simply go ahead and orbifold $V_L$ by $\hat G$. For our purposes however this is problematic. The issue is that in this case the identity automorphism in $G$ has more than one preimage. By (\ref{vacanomaly}) the vacuum anomaly of the $\hat g$-twisted module only depends on $g$ and is thus zero if $g$ is trivial. This means that $\hat g$ potentially introduces many light states, which is undesirable for our construction of CFTs with large gap. We therefore want to establish that for $E(m)$, we indeed have $\hat E(m) \simeq E(m)$. This is the purpose of the following theorems and the lemmas in appendix~\ref{ss:appAut}.

From (\ref{Bg}) we see immediately that $u_z(\alpha)=1$. For the other generators, we have the following theorems:
\begin{thm}\label{Emeven}
Let $m\geq 4$ be even. We then have 
\be
u_{\Si{ar}}(\alpha) = u_z(\alpha) = 1 \qquad a=1,2,\ldots m\ ,
\ee
for all $\alpha$.
\end{thm}
For a proof of this see lemma~\ref{lem:ueven} in the appendix. For $m$ odd, this is no longer quite true: the $u_{\Si{1r}}$ are in general non-trivial. However, for our purposes the following weaker result is sufficient:
\begin{prop}\label{prop:uinv}
Let $m\geq 3$. We have $u_{\Si{ar}}(\alpha \Si{bs})=u_{\Si{ar}}(\alpha)$ for all generators. 
\end{prop}
For $m\geq 4$ this is the statement of proposition~\ref{prop:appcov} in the appendix, and for $m=3$ it follows from a straightforward computation. This leads to the main theorem of this section:

\begin{thm}\label{thm:Elift}
Let $m\geq 3$. Then the generators $\lSi{ar}$ continue to satisfy the relations of the extraspecial group $E(m)$:
\be\label{Emliftrelations}
[ \hat z, \lSi{ar}]=1 \qquad  \hat z^2 =  \lSi{ar}^2=1 \qquad [ \lSi{ar},  \lSi{bs}]= \hat z^{\delta_{a,b}\delta_{r,s+1\mod 2}}
\ee
That is, $\hat E(m) \cong E(m)$.
\end{thm}
\begin{proof}
Since the original generators $\Si{ar}$ satisfy the relations and $u_z(\alpha)=1$, it is enough to check the relations
\be
u_{\Si{ar}}(\alpha\Si{ar})u_{\Si{ar}}(\alpha)=1
\ee
and
\be
u_{\Si{ar}}(\alpha\Si{bs})u_{\Si{bs}}(\alpha)=
u_{\Si{bs}}(\alpha\Si{ar})u_{\Si{ar}}(\alpha)\ .
\ee
This however is an immediate consequence of proposition~\ref{prop:uinv} and the fact that $u_{\Si{ar}}^2=1$.
\end{proof}

\section{Defect representation and intertwiners}\label{s:defrep}
\subsection{Intertwining representation and Schur averaging}\label{ss:intertwining}
Let $H<Aut(N)$ and $\rho$ be an irreducible projective representation on $U$ of a group $N$ with 2-cocycle $P\in H^2(N,U(1))$. Then for $h\in H$, $\rho_h:=\rho\circ h$ is again an irreducible projective representation on $U$ with say 2-cocycle $P_h$. 
Assume that $\rho$ and $\rho_h$ are projectively equivalent, that is there is an automorphism $T_h$ on $U$ such that
\be
T_h \rho_h(\alpha) T_h^{-1} = B_h(\alpha) \rho(\alpha) \qquad \forall \alpha \in N
\ee
for a scalar function $B_h$ such that $PP_h^{-1} = dB_h$.

Next note that we have
\be
(T_{gh})^{-1}T_gT_h \rho(\alpha) T_h^{-1} T_g^{-1} T_{gh} = B_h^{-1}(\alpha) B_g^{-1}(\alpha h) B_{gh}(\alpha) \rho(\alpha)\ .
\ee
From Schur's lemma it follows that if
\be\label{Bcondition}
B_h(\alpha) B_g(\alpha h) B_{gh}^{-1}(\alpha)=1
\ee
then we have
\be
T_g T_h = \lambda(g,h)T_{gh}
\ee
for some scalar $\lambda(g,h)$.
In particular, if (\ref{Bcondition}) is satisfied for all $g,h\in H$, then the $T_g$ give a projective representation of $H$ on $U$. We call this representation the \emph{intertwining representation} of $H$. In practice, we will construct the $T_g$ for the generators of $H$ and then check (\ref{Bcondition}) for all relations of the generators.

To find $T$ we will use Schur averaging: Let $P$ and $P_h$ be the 2-cocycles of the representations $\rho$ and $\rho_h$. Since they are assumed to be projectively equivalent, we can find a function $B_h$ such that $PP_h^{-1}=dB_h$, giving the sought after coboundary. Let $t$ be a linear map. We can then take the Schur average of $t$ as
\be\label{Schuraverage}
T_h:=\frac1{\sqrt{|N|}} \sum_{\alpha\in N} B_h(\alpha)^{-1} \rho(\alpha)^{-1} t \rho_h(\alpha)
\ee
to obtain a projectively $N$-linear map $T_h$, that is $T_h\rho_h T_h^{-1}=B_h \rho$. Here we chose the normalization factor $1/\sqrt{|N|}$ for future convenience. 
Since $\rho$ is irreducible, unless $T_h$ vanishes it is an isomorphism.

We want to apply this to the defect representation $(\rho,\Omega_0)$ and VOA automorphisms $\hat h$. However, note that the results in section~\ref{ss:Elift}  allow us to work with unlifted lattice automorphisms $h$ rather than their lifts $\hat h$: We have $B_{\hat h}(\alpha)=u_h(\alpha)B_h(\alpha)$. Because the $u_h(\alpha)$ are compatible with the relations of $E(m)$, if follows that the $B_{\hat h}(\alpha)$ will satisfy (\ref{Bcondition}) if the $B_h(\alpha)$ do.
Moreover we have $B_{\hat h}(\alpha)^{-1} \rho_{\hat h}(\alpha) = u_h(\alpha)^{-1}B_h(\alpha)^{-1} u_h(\alpha) \rho_h(\alpha)=B_h(\alpha)^{-1} \rho_h(\alpha)$,
so that the Schur average (\ref{Schuraverage}) is $T_{\hat h}=T_h$.

\subsection{Coordinate notation}\label{ss:coordinate}
Let us now work out explicit expressions for the Schur average (\ref{Schuraverage}). It turns out that our group $N$ will always be of the form $\Z_2^{2d}$ so that we can work over $\F_2$. We pick a symplectic basis $\{\alpha_i,\beta_i\}$ for $N$, that is a basis that is symplectic with respect to the form $C(\cdot,\cdot)$ introduced in section~\ref{ss:dr}.
In coordinate notation we write an element $\alpha \in N$ as
\be
\alpha = \sum_{i=1}^{d} a_i \alpha_i + \check a_i \beta_i
\ee
and define the coordinate vector $a=(\co{a},\ct{a})$ with $\co{a}=(a_1,\ldots , a_{d})$ and $\ct{a}=(\check a_1,\ldots, \check a_{d})\in\F_2^d$. The idea is then to work in the exponent of $(-1)$ and reduce everything to a problem in linear algebra over $\F_2$.

From section~\ref{ss:dr}, $\Omega_0$ is spanned by the basis vectors $e_{\ct{c}}$ with $\ct c \in \F_2^d$. The representation $\rho$ acts as
\be
\rho(\alpha)(e_\ct{c})= (-1)^{\co{a}\cdot \ct{c}}e_{\ct{c}+\ct{a}}
= (-1)^{a R_{1/2}(\ct{c}\iota_2)^T}e_{\ct{c}+a\pi_2}
\ee
where $\pi_2$ is the projection $a\pi_2=\ct{a}$, $\iota_2$ is the inclusion $\ct{a}\iota_2=(\underline{0},\ct{a})$ and
\be
R_{1/2}= \begin{pmatrix}0&I\\0&0 \end{pmatrix}\ .
\ee
Note again that we take $a$ to be row vectors so that operators act from the right.
A quick computation shows $\rho(\alpha)\circ \rho(\alpha)(e_\ct{c})=(-1)^{aR_{1/2}a^T}e_\ct{c}$, so that
\be
\rho(\alpha)^{-1}= (-1)^{aR_{1/2}a^T}\rho(\alpha)\ . 
\ee
The cocycle of $\rho$ is given by
\be
P(\alpha,\beta)=(-1)^{\co{a}\cdot\ct{b}}=(-1)^{a R_{1/2} b}\ .
\ee
Similar expressions hold for $\rho_h$. In particular we have
\be
P(\alpha,\beta)P_h(\alpha,\beta)^{-1} = (-1)^{aR_{1/2}b + a\hm R_{1/2} \hm^T b^T} = (-1)^{a \Delta^h b^T}
\ee
where $\hm$ is the matrix of $h$ in the basis given above and we defined the matrix
\be
\Delta^h = R_{1/2} + \hm R_{1/2}\hm^T\ .
\ee
It turns out that for all cases we considered, this matrix is symmetric and has zeros on the diagonal. 
To construct $B_h$, we can thus use the following observation:
Let $D: \Z_n^m \times \Z_n^m \to \Z_n$ be a symmetric bilinear form with vanishing diagonal entries, that is $D(\alpha,\beta)= \sum_{i,j}c_{ij} \alpha^i\beta^j$ with $c_{ij}=c_{ji}$ and $c_{ii}=0$. Then $\xi_n^{D(\alpha,\beta)} = B(\alpha+\beta)B(\alpha)^{-1}B(\beta)^{-1}$ with $B(\alpha) = \xi_n^{\sum_{i>j}c_{ij}\alpha^i\alpha^j}$. As a cocycle, $\xi_n^D$ is thus trivial with $\xi_n^D= dB$.

In view of (\ref{Bcondition}) we want to take $B_h$ slightly more general: we want to allow for a linear part in the exponent, whose derivative vanishes because of linearity and thus leaves the 2-cocycle unchanged. We thus take the ansatz
\be
B_h(a) := (-1)^{a \Delta^h_{1/2}a^T+a s_h^T}
\ee
where $s_h$ is some vector and we introduced the upper triangular matrix
\be
(\Delta^h_{1/2})_{ij} = \left\{ \begin{array}{ccl} \Delta^h_{ij} &:& i<j\\
0&:& \textrm{else} \end{array}\right.\ .
\ee

\subsection{Intertwining representation of $E(m)$}
We now want to use the Schur average intertwiner operators (\ref{Schuraverage}) to construct a projective representation of $E(m)$ on $\Omega_0$. We proceed as follows: We will first construct intertwiner matrices $T_h$ for all $2m+1$ generators of $E(m)$. We then check that these indeed generate a projective representation by checking that they satisfy the projective version of the relations (\ref{Emliftrelations}). 

For the generators $\lSi{ar}$ we construct intertwiners $T_{ar}$ using Schur averaging.
For $\hat z$ we observe that it acts trivially on $N$ and thus take $T_{\hat z}=\id$. The relations involving $T_{\hat z}$ are thus immediately satisfied, so that we only need to check the projective versions of the remaining relations, namely
\be\label{projCond}
T_{ar}^2= \lambda(a,r)\id \ , \qquad T_{ar} T_{bs} = \mu(a,r;b,s)T_{bs}T_{ar}\ 
\ee
for some scalars $\lambda$ and $\mu$. To establish this we use the comment after (\ref{Bcondition}); namely if
\be\label{Bordcond}
B_h(\alpha)B_h(\alpha h) = 1
\ee
and
\be\label{Bcommcond}
B_h(\alpha)B_g(\alpha h) = B_g(\alpha) B_h(\alpha g)\ ,
\ee
then (\ref{projCond}) are satisfied and the $T_h$ indeed form a projective representation of $E(m)$.

The task is thus to construct such $B_h$ by choosing $s_h$ appropriately. The two conditions above become
\be
a(\Delta^h_{1/2}+\hm \Delta^h_{1/2} \hm^T)a^T = a(s_h^T+\hm s_h^T)\ 
\ee
and
\be
a(\Delta^h_{1/2}+\hm \Delta^g_{1/2} \hm^T + \Delta^g_{1/2}+\gm \Delta^h_{1/2} \gm^T)a^T= a (s_h^T+\hm s_g^T + s_g^T+\gm s_h^T)\ 
\ee
for all pairs of generators $h,g$. We note that we only have a chance of satisfying this equation if the only terms of quadratic form on the left hand side are diagonal, that is if the matrix on the left hand side is symmetric. In that case we can use the fact that we work in characteristic 2 to replace $a_i^2$ by $a_i$ and turn quadratic terms $a Q a^T$ into linear terms $a q$, where $q$ is the diagonal of the quadratic form $Q$. Each condition then leads to a system of linear equations. For all examples that we consider, this linear system does have a solution for the $s_h$, allowing us to construct a projective representation of $E(m)$ on $\Omega_0$.

We collect all these expressions and plug them into the Schur average (\ref{Schuraverage}). It only remains to choose the linear map $t$. The simplest choice is to take the identity, which gives the intertwiner matrix
\be
T_h(e_\ct{c})= \frac1{\sqrt{|N|}}\sum_{a\in N} 
(-1)^{a\Delta^h_{1/2}a^T+a(\hm+I)R_{1/2}(\ct{c}\iota_2)^T+a R_{1/2}(I+\hm^T)a^T+as_h^T}
e_{\ct{c}+a(\hm +I)\pi_2}\ .
\ee
However, it turns out that sometimes the Schur average of the identity vanishes. In general it is thus better to work with elementary matrices, that is maps $t_{\ct{e}\ct{f}}$ such that $t_{\ct{e}\ct{f}}(e_\ct{c})=\delta_{\ct{c}\ct{e}}e_\ct{f}$. This gives the intertwiner matrix
\be\label{Telem}
T^{\ct{e}\ct{f}}_h(e_\ct{c})= \frac1{\sqrt{|N|}} \sum_{\substack{a\in N :\\ \ct{c}+a\hm\pi_2=\ct{e}}} 
(-1)^{a\Delta^h_{1/2}a^T+a  R_{1/2}a^T+ a\hm R_{1/2}(\ct{c}\iota_2)^T+ aR_{1/2}(\ct{f}\iota_2)^T+as_h^T}
e_{\ct{f}+a\pi_2}\ .
\ee
Evaluating such sums is straightforward in principle. In practice however the group $N$ very quickly becomes extremely big, so that summing by brute force becomes impossible. In appendix~\ref{app:SchurAve} we give an algorithm that gives an explicit expression for the matrix entries and allows us to work with such matrices without the need to do any summation. The overall idea is to work with $N$ as a vector space over $\F_2$ and use linear algebra. These computations then scale exponentially better than when working with group elements of $N$ individually. The algorithm allows us to implement our computations up to $m=7$ in Magma \cite{Mag97} or Mathematica \cite{Mathematica}. We attach the code as a supplement to this article.

\subsection{Intertwining representations for $m=3,5,7$}\label{ss:intertw}

We will now construct the intertwiner representations for $m=3,5,7$ for the three non-trivial conjugacy classes in proposition~\ref{prop:Emcc}, namely $g = z, \Si{11}, \Si{11}\Si{12}$. 

We start with $g=z =-\id$. We have $L_g^\perp = L$, so that $N= L/(2L) \simeq (\Z/2\Z)^{2^m}$. We therefore have $|\check A|= 2^{2^{m-1}}= |\Omega_0|$. It also follows that $z$ acts trivially on $N$, so that $T_{z}= \id$. This means that the intertwiner representation that we construct is actually a representation of $E(m)/\langle z\rangle \simeq \Z_2^{2m}$. 

We will now explicitly construct the intertwiner representation for $m=3,5,7$. This projective representation will have a cocycle in $H^2(\Z_2^{2m},U(1))$. 
The aspect we are most interested in is if this cocycle is trivial, that is if the intertwiner representation is in fact a linear representation. To check this we check if the generators satisfy the relations of $\Z_2^{2m}$,
\bea
T_{ar}^2&=& \id \ , \label{ESord} \\
T_{ar} T_{bs} &=& T_{bs}T_{ar} \label{EScomm}\ .
\eea

Let us start with $m=3$.  To compute the intertwiner matrices $T_{ar}$ for the  generators $\Si{ar}$ we take the Schur average of the elementary matrix with $\ct{e}=\ct{f}= 0$. After choosing an appropriate normalization, this gives intertwiner matrices $T_{ar}$ that act in the following way on the basis vectors:
\bea
T_{11}(e_\ct{c}) &=& e_{(\check{c}_1,\check c_2,\check c_3,\check c_4+\check c_2+\check c_3)}\\
T_{12}(e_\ct{c}) &=& (-1)^{\check c_1 (\check c_2+\check c_3)} e_{(\check{c}_1,\check c_2+\check c_1,\check c_3+\check c_1,\check c_4)}\\
T_{21}(e_\ct{c}) &=& e_{(\check{c}_1,\check c_3,\check c_2,\check c_4)} \\
T_{22}(e_\ct{c}) &=& e_{(\check{c}_1,\check c_2,\check c_3,\check c_4+\check c_1)} \\
T_{31}(e_\ct{c}) &=& \frac12(e_{(\check{c}_1,\check c_2,\check c_3,\check c_4)}+e_{(\check{c}_1,\check c_2,\check c_3,\check c_4+1)}) 
+ \frac12(-1)^{\check c_2+\check c_3}(e_{(\check{c}_1,\check c_2+1,\check c_3+1,\check c_4)}-e_{(\check{c}_1,\check c_2+1,\check c_3+1,\check c_4+1)})\\
T_{32}(e_\ct{c}) &=& (-1)^{\check c_1(\check c_2+\check c_3)}e_{(\check{c}_1,\check c_2,\check c_3,\check c_4)}
\eea
It is straightforward to check that they satisfy (\ref{ESord}) and (\ref{EScomm}). The 2-cocycle is thus indeed trivial.

Next we turn to $m=5$. 
The intertwiner matrices $T_{ar}$ are too big to be given explicitly and to check the group relations explicitly. However,  as long as we ensure that $B$ satisfies (\ref{Bordcond}) and (\ref{Bcommcond}), Schur's lemma allows us to only compare a single non-vanishing entry for each of the relations. This is enough to extract the scalars $\lambda$ and $\mu$ in (\ref{projCond}).

We use linear algebra to find the $s_h$ for each of the generators. We give more details for the construction including our choice of elementary matrices and the appropriate
normalizations in appendix~\ref{ss:intm5}. We find the following relation matrix for values for a specific entry of the products of the 10 generator matrices $T_{ar}T_{bs}$, where the ordering is first by $a$ and then by $r$:
\be\label{commMatm5}
\left(
\begin{array}{cccccccccc}
 1 & \frac{1}{4} & 1 & \frac{1}{8} & \frac{1}{4} & \frac{1}{2} & \frac{1}{2} & \frac{1}{2} & \frac{1}{2} & \frac{1}{2} \\
 \frac{1}{4} & 1 & \frac{1}{4} & \frac{1}{4} & \frac{1}{8} & \frac{1}{4} & \frac{1}{8} & \frac{1}{4} & \frac{1}{4} & \frac{1}{4} \\
 1 & \frac{1}{4} & 1 & \frac{1}{8} & \frac{1}{4} & \frac{1}{2} & \frac{1}{2} & \frac{1}{2} & \frac{1}{2} & \frac{1}{2} \\
 \frac{1}{8} & \frac{1}{4} & \frac{1}{8} & 1 & \frac{1}{16} & \frac{1}{8} & \frac{1}{8} & \frac{1}{8} & \frac{1}{8} & \frac{1}{8} \\
 \frac{1}{4} & \frac{1}{8} & \frac{1}{4} & \frac{1}{16} & 1 & \frac{1}{8} & \frac{1}{4} & \frac{1}{4} & \frac{1}{4} & \frac{1}{4} \\
 \frac{1}{2} & \frac{1}{4} & \frac{1}{2} & \frac{1}{8} & \frac{1}{8} & 1 & \frac{1}{2} & \frac{1}{2} & \frac{1}{2} & \frac{1}{2} \\
 \frac{1}{2} & \frac{1}{8} & \frac{1}{2} & \frac{1}{8} & \frac{1}{4} & \frac{1}{2} & 1 & \frac{1}{4} & \frac{1}{4} & \frac{1}{2} \\
 \frac{1}{2} & \frac{1}{4} & \frac{1}{2} & \frac{1}{8} & \frac{1}{4} & \frac{1}{2} & \frac{1}{4} & 1 & \frac{1}{2} & \frac{1}{4} \\
 \frac{1}{2} & \frac{1}{4} & \frac{1}{2} & \frac{1}{8} & \frac{1}{4} & \frac{1}{2} & \frac{1}{4} & \frac{1}{2} & 1 & \frac{1}{4} \\
 \frac{1}{2} & \frac{1}{4} & \frac{1}{2} & \frac{1}{8} & \frac{1}{4} & \frac{1}{2} & \frac{1}{2} & \frac{1}{4} & \frac{1}{4} & 1 \\
\end{array}
\right)
\ee
The fact that the diagonal entries are 1 establishes (\ref{ESord}) and the fact the matrix is symmetric establishes (\ref{EScomm}). The 2-cocycle is thus again trivial.

For $m=7$ we again give details for the construction in appendix~\ref{ss:intm7}. The relation matrix is given by
\be
\left(
\begin{array}{cccccccccccccc}
 1 & 2^{-8} & 1 & 2^{-10} & 2^{-1} & 2^{-4} & 2^{-3} & 2^{-2} & 2^{-2} & 2^{-4} & 2^{-3} & 2^{-4} & 2^{-6} & 1 \\
 2^{-8} & 1 & 2^{-8} & 2^{-7} & 2^{-8} & 2^{-9} & 2^{-9} & 2^{-8} & 2^{-9} & 2^{-9} & 2^{-9} & 2^{-8} & 2^{-9} & 2^{-8} \\
 1 & 2^{-8} & 1 & 2^{-10} & 2^{-1} & 2^{-4} & 2^{-3} & 2^{-2} & 2^{-2} & 2^{-4} & 2^{-3} & 2^{-4} & 2^{-6} & 1 \\
 2^{-10} & 2^{-7} & 2^{-10} & 1 & 2^{-10} & 2^{-11} & 2^{-11} & 2^{-10} & 2^{-11} & 2^{-9} & 2^{-11} & 2^{-10} & 2^{-11} & 2^{-10} \\
 2^{-1} & 2^{-8} & 2^{-1} & 2^{-10} & 1 & 2^{-5} & 2^{-3} & 2^{-3} & 2^{-2} & 2^{-5} & 2^{-3} & 2^{-5} & 2^{-6} & 2^{-1} \\
 2^{-4} & 2^{-9} & 2^{-4} & 2^{-11} & 2^{-5} & 1 & 2^{-5} & 2^{-5} & 2^{-5} & 2^{-6} & 2^{-6} & 2^{-6} & 2^{-6} & 2^{-4} \\
 2^{-3} & 2^{-9} & 2^{-3} & 2^{-11} & 2^{-3} & 2^{-5} & 1 & 2^{-5} & 2^{-3} & 2^{-6} & 2^{-3} & 2^{-5} & 2^{-6} & 2^{-3} \\
 2^{-2} & 2^{-8} & 2^{-2} & 2^{-10} & 2^{-3} & 2^{-5} & 2^{-5} & 1 & 2^{-4} & 2^{-4} & 2^{-3} & 2^{-5} & 2^{-6} & 2^{-2} \\
 2^{-2} & 2^{-9} & 2^{-2} & 2^{-11} & 2^{-2} & 2^{-5} & 2^{-3} & 2^{-4} & 1 & 2^{-6} & 2^{-3} & 2^{-4} & 2^{-6} & 2^{-2} \\
 2^{-4} & 2^{-9} & 2^{-4} & 2^{-9} & 2^{-5} & 2^{-6} & 2^{-6} & 2^{-4} & 2^{-6} & 1 & 2^{-4} & 2^{-5} & 2^{-7} & 2^{-4} \\
 2^{-3} & 2^{-9} & 2^{-3} & 2^{-11} & 2^{-3} & 2^{-6} & 2^{-3} & 2^{-3} & 2^{-3} & 2^{-4} & 1 & 2^{-7} & 2^{-6} & 2^{-3} \\
 2^{-4} & 2^{-8} & 2^{-4} & 2^{-10} & 2^{-5} & 2^{-6} & 2^{-5} & 2^{-5} & 2^{-4} & 2^{-5} & 2^{-7} & 1 & 2^{-7} & 2^{-4} \\
 2^{-6} & 2^{-9} & 2^{-6} & 2^{-11} & 2^{-6} & 2^{-6} & 2^{-6} & 2^{-6} & 2^{-6} & 2^{-7} & 2^{-6} & 2^{-7} & 1 & 2^{-6} \\
 1 & 2^{-8} & 1 & 2^{-10} & 2^{-1} & 2^{-4} & 2^{-3} & 2^{-2} & 2^{-2} & 2^{-4} & 2^{-3} & 2^{-4} & 2^{-6} & 1 \\
\end{array}
\right)
\ee

The fact that the diagonal entries are 1 establishes (\ref{ESord}) and the fact the matrix is symmetric establishes (\ref{EScomm}). The 2-cocycle of the intertwiner representation is thus again trivial.

Finally let us turn to the defect representations of the other two conjugacy classes. For $g=\Si{11}$, (\ref{Si11Si12form}) leads to a Smith Normal Form of $\diag( 2 I_{2^{m-2}}, I_{2^{m-2}},0,\ldots 0)$. This gives the group $N = (\Z/2\Z)^{2^{m-2}}$.
For $g=\Si{11}\Si{12}$, the SNF is $\diag( 2 I_{2^{m-2}}, 2 I_{2^{m-2}}, I_{2^{m-2}},I_{2^{m-2}})$, giving the group $N=(\Z/2\Z)^{2^{m-1}}$. Somewhat surprisingly, we found that for $m=3,5,7$ the centralizer groups $\Z_2\times E(m-1)$ and $\Z_4\circ E(m-1)$ act trivially on $N$. Their intertwiner representations are thus trivial and automatically have vanishing 2-cocycle.

Our overall conclusion is thus that for $m=3,5,7$, the representations $\phi_g$ of $C_{E(m)}(g)$ on the $g$-twisted modules have vanishing 2-cocycle and are thus linear.

\section{Extraspecial Orbifolds of Barnes-Wall lattice VOAs}\label{s:orbifolds}
\subsection{Large gap CFTs from lattice orbifolds}
Before collecting our results to construct our extraspecial orbifolds, let us give some general considerations for constructing CFTs with large gap from lattice orbifolds. In particular, let us discuss what types of lattices and orbifold groups are appropriate for this purpose.

The overall goal is to avoid any light states. First we need to ensure that no such states appear from lattice vectors. This means we should work with lattices whose shortest vector has large length. For this we would ideally like to work with extremal lattices whose shortest vectors have length square $2(\lfloor d/24\rfloor+1)$. Unfortunately, such lattices have only been constructed up to $d=80$ \cite{MR209983,MR1662447,MR1489922,MR2999133,MR3225314,MR1985228}, and they are known not to exist if $d>163264$ \cite{MR0376536,MR2854563}. For $d=48$ and $d=72$ their orbifolds were considered in \cite{Gemunden:2018mkh,Gemunden:2019dtr}, but no CFTs with gap larger than 2 were found. Barnes-Wall lattices are therefore the next best candidates. For odd $m$ they are even and unimodular, thus giving holomorphic lattice VOAs, and their shortest vectors have length squared $\sqrt{d/2}$. Moreover they have large automorphism groups $BRW(m)$, which becomes important in the next step.

Next we need to eliminate light free boson descendants. To this end we want to orbifold by some group $G$ of automorphisms. There are two considerations here. First, we want to ensure that the twisted modules that are introduced in the orbifold do not produce new light states. To avoid that we want make sure that the vacuum anomalies $\rho_g$ of the $g$-twisted sector is large. As we can see from formula (\ref{vacanomaly}), this constrains the eigenvalues of the elements of $G$.
Second, we want to project out as many of the free boson states in the untwisted sector as possible. In the language of section~\ref{ss:Efacts}, we want the trivial representation of $G$ to appear as few times as possible in the Heisenberg representation. On the one hand let us note that if $G$ is Abelian, then \cite{Gemunden:2019dtr}
\be
\dim \Hr_{(2)}^G \geq \frac d2\ , \qquad \dim \Hr_{(3)}^G \geq d\ .
\ee
To construct a CFT with gap larger than 2 it is thus necessary to consider non-Abelian orbifold groups.
On the other hand we can also bound the best possible case. Note that $G<O(d,\R)$, so that any state invariant under the orthogonal group is also invariant under $G$. We want any light state that is invariant to be a Virasoro descendant. In table~\ref{t:Oinv} we list the lightest $O(d,\R)$ invariant states and Virasoro descendants. We see that up to $h=3$ the two match up. Starting at $h=4$ however, there are more invariant states than Virasoro descendants.
This means that there is necessarily a Virasoro highest weight state at weight 4. Lattice orbifolds can thus at best produce a VOA of gap 4. This explains why we content ourselves with just the extraspecial group $E(7)$ and did not attempt to consider larger automorphism groups.

\begin{table}[h]
\begin{tabular}{c|c|c}
$h$&$O(d,\R)$ invariants& Virasoro descendants\\
\hline
0 & $\vac$ & $\vac$\\
1 & -- & --\\
2 & $a_{-1}\cdot a_{-1}\vac$ & $L_{-2}\vac$\\
3 & $a_{-1}\cdot a_{-2}\vac$ & $L_{-3}\vac$\\
4 & $a_{-2}\cdot a_{-2}\vac, a_{-1}\cdot a_{-3}\vac, (a_{-1}\cdot a_{-1})^2\vac$ & $L_{-4}\vac, L_{-2}L_{-2}\vac$
\end{tabular}
\caption{$O(d,\R)$ invariant states and Virasoro descendants at low weight \label{t:Oinv}}
\end{table}

Finally there is the issue of constructing a CFT from the modules of the VOA $V^G$. We can always choose the diagonal invariant. However, the resulting CFT will always contain a Virasoro primary state of weight $(h,\bar h)=(1,1)$, namely the state $a(-1)\cdot \bar a(-1)\vac$. To avoid this, we can try to construct a holomorphic VOA by finding a modular invariant that extends $V^G$ to a holomorphic VOA $V^{orb(G)}$. The problem here is that this requires the anomaly $\omega\in H^3(G,U(1))$ to be trivial, which imposes constraints on $G$.

\subsection{The diagonal orbifold}
Let us now apply the results of the previous sections to construct extraspecial orbifolds of Barnes-Wall lattice VOAs. By theorem~\ref{thm:Elift} we can lift $E(m)$ to a group of automorphisms of the lattice VOA $V_L$ isomorphic to $E(m)$, allowing us to indeed  consider extraspecial orbifolds of $V_L$. We will focus on the unimodular lattices with $m$ odd. In that case the VOA of the BW lattice is holomorphic and the expressions in section~\ref{ss:oftheory} apply. We will mostly discuss the instances $m=3,5,7$. In principle our methods also apply to higher cases, but obtaining explicit expressions becomes computationally harder and harder.

We will be interested in the spectrum of light states of the full CFT. That is, we want to evaluate the character of the CFT
\be\label{Hchar}
\Tr_{\cH} q^{L_0-c/24} \bar q^{\bar L_0-c/24}
\ee
and consider its first few terms. To this end, we need expressions for the characters of the twisted representations. To do this, we can use the fact that we already constructed the intertwiner representation $T_h$ on the defect representation $\Omega_0$. Following \cite{Gemunden:2019dtr} we could lift this to the representation $\phi_g(\cdot)$ on $W_g$ and then use their expressions for the twisted twining characters to obtain (\ref{Hchar}). For the specific question we are considering here it turns out that this is unnecessary. It is enough to use the following  facts about the structure of the characters.

The characters of the irreducible representations can best be obtained from the twisted twining characters 
\be
T(g,h;\tau):=\Tr_{W_g}\phi_g(h) q^{L_0-c/24}
\ee
For a lattice VOA orbifold, the twisted twining characters consist of two factors,
\be
T(g,h;\tau) = T_H(g,h;\tau) T_L(g,h;\tau)\ .
\ee
Explicit expressions for $T_H$ and $T_L$ can be found \eg in section~6 of \cite{Gemunden:2019dtr}. The Heisenberg factor $T_H$ takes into account the contribution of the free boson descendants. It is essentially a modified eta function. What is important for our purposes is that its leading term is $q^{\rho_g}$, where $\rho_g$ is the vacuum anomaly defined in (\ref{vacanomaly}) and evaluated in proposition~\ref{prop:Emcc}. The lattice factor $T_L$ is the contribution of the lattice vectors. It is essentially a theta function of a sublattice of $L$. 

Let us focus on the BW lattice with $m=7$.
We are interested in states of total weight up to 3, meaning we can expand characters up to terms of total order 3. From proposition~\ref{prop:Emcc} we see that for all non-trivial elements $g$, the twisted sector vacuum anomaly satisfies $\rho_g\geq 4$. Just purely from the Heisenberg part, it follows that the characters of the twisted modules do not make any contributions to the first 3 orders in the expansion and can thus be ignored. It this thus enough to focus on the untwisted sector.

In the untwisted sector, $T_L(g,h;\tau)$ is the lattice theta function of the fixed point lattice of $h$ (up to phases in front of the terms). In particular, since the shortest vector of the $m=7$ BW lattice has length squared $8$, they are of the form 
\be
1 + O(q^4)\ .
\ee
It follows that the only contribution to the character up to weight 3 comes from the free boson descendants in the untwisted module. But for the untwisted module we have that the free boson descendants give the Heisenberg representation $F$ introduced in definition~\ref{defn:Hr}. The character is thus given by
\be\label{T1char}
T(1,h;\tau) = \chi^\Hr_h(q) + O(q^4)\ .
\ee

\begin{prop}\label{prop:diag}
The diagonal invariant of the $E(7)$-orbifold gives a non-holomorphic full CFT with
\be
\Tr_{\cH} q^{L_0}\qb^{\bar L_0}= 1 + q^2 +q\qb +\qb^2 + q^3+\qb^3 + q\qb^2+q^2\qb + O(|q|^4)
\ee
where $O(|q|^4)$ denotes terms $q^a \qb^b$ with total weight $a+b \geq 4$.
\end{prop}
\begin{proof}
We start from expression (\ref{Hdiag}) for $\cH$.
From the comments above we know that to this order only the untwisted characters contribute, and for those in turn only the Heisenberg character $\chi^\Hr$ contributes. To establish the proposition, it is thus enough to compute how often the trivial representation appears in the appropriate products $\Hr_{(i)}\otimes\Hr_{(j)}$ of $\Hr_{(0)}, \Hr_{(1)}, \Hr_{(2)}$ and $\Hr_{(3)}$. The proposition then follows immediately from proposition~\ref{prop:HeInvar}.
\end{proof}

\subsection{The holomorphic orbifold}
Let us now discuss the possibility of constructing holomorphic extraspecial orbifolds.
By the comments in section~\ref{ss:oftheory}, this depends on the anomaly $\omega \in H^3(G,U(1))$. For the case of our extraspecial groups the cohomology groups were computed in \cite{MR3990846}. For $m\geq 3$ they are given by\footnote{There is a typo in the expression for $n$ in \cite{MR3990846}. We thank Theo Johnson-Freyd for communications on this point.}
\be
H^3(E(m),U(1)) \simeq \Z_2^n \times \Z_4 \qquad \textrm{where}\ n = \binom{2m}{2}+\binom{2m}{3}-1\ .
\ee
A holomorphic extension only exists if $\omega$ is trivial \cite{Evans:2018qgz}. 
In principle we could try to reconstruct $\omega$ from the modular transformation data of the module characters. We do not attempt to do that. Instead we check two properties.

First we check that $\omega$ restricted to cyclic subgroups of $E(m)$ is trivial. There is a simple criterion for when the restriction of $\omega$ to a cyclic subgroup is trivial \cite{vanEkeren2017}: We define the \emph{type} $t_{\hat g}$ of an automorphism $\hat g$ of order $n$ as
\be
t_{\hat g} = n^2 \rho_{\hat g} \mod n\ ,
\ee
where $\rho_{\hat g}$ is the conformal weight of the $\hat g$-twisted module $W_{\hat g}$. Then the restriction of $\omega$ to the cyclic group $\langle \hat g\rangle$ is trivial iff $t_{\hat g}=0$.
\cite{vanEkeren2017} call this the `type 0 condition', and the physics literature calls it `level matching'\cite{Vafa:1986wx}. If $\hat g$ is a standard lift of $g$, then $\rho_{\hat g}=\rho_g$ as defined in (\ref{vacanomaly}).

For $E(7)$, we checked that all $\hat g$ are indeed standard lifts, that is $u_g=1$ on $L_g$. This is automatically true for $z$ and all elements that have at least one factor of the form $\sigma_1\sigma_2$ because their fixed point lattices are trivial. For the remaining elements, in view of lemma~\ref{lem:utriv} and proposition~\ref{prop:uinv}, it is enough to check that $u_{\Si{11}}=1$ on the fixed point lattice for all elements of order 2 that have $\sigma_1$ in the first factor, and similar for $u_{\Si{12}}$. We checked that this is the case using Mathematica; the notebook can be found in the supplements.
Since then $\rho_{\hat g}=\rho_g$, it follows from the expressions in proposition~\ref{prop:Emcc} that the type 0 condition is indeed satisfied.

Second, our explicit computations in section~\ref{s:defrep} show that the representations $\phi_g$ on the $g$-twisted modules are indeed linear for $m=3,5,7$. Because of the relation (\ref{descend}) between $\omega$ and the $c_g(\cdot,\cdot)$ this is further evidence that $\omega$ is trivial and that we can construct holomorphic extraspecial orbifold VOAs $V^{orb(E(m))}$. 

Based on these two pieces of evidence, let us therefore assume that $\omega$ indeed vanishes for $m=7$. We can then use (\ref{VorbG}) to write  the character of $V^{orb(E(7))}$. By the same argument as above, up to order $q^3$ the twisted modules and the lattice states in the untwisted module do not contribute. We thus have 
\begin{multline}
\Tr_{V^{orb(E(7))}} q^{L_0} = \Tr_{V^{E(7)}} q^{L_0} +O(q^4) = 
\frac{1}{|E(7)|}\sum_{h\in E(7)}\chi^\Hr_h(q) +O(q^4)\\
= 1 + q^2 +q^3 +O(q^4)\ ,
\end{multline}
where we used (\ref{T1char}) and proposition~\ref{prop:HeInvar}.
We note that the $q^2$ term corresponds to the conformal vector $L_{-2}\vac$ and the $q^3$ to the vector $L_{-3}\vac$. This VOA thus has gap 4.
Collecting all this gives the following conjecture:
\begin{con} The holomorphic $E(7)$-orbifold of the 128-dimensional Barnes-Wall lattice VOA exists. It is a holomorphic CFT of central charge $c=128$ with gap 4.
\end{con}
To our knowledge this is the first construction of a holomorphic VOA (and indeed of any full CFT) with gap bigger than 2. Note that it is not an extremal VOA though, as its central charge is larger than $c=72$, as would be required for an extremal VOA with gap 4.

\section{Conclusion}\label{s:conclusion}
In this section we give a less technical summary of the conjectured CFT of gap 4, its construction and its known properties. We also speculate on some of its properties and discuss how one could establish them. The section is aimed at readers who are interested in using this CFT, but do not want to go through the details of its construction.

\subsection{The construction}
We constructed a holomorphic CFT by generalizing the construction of the Monster CFT. We started with a holomorphic lattice CFT obtained by compactifying 128 free bosons on the Barnes-Wall lattice $BW(7)$ of dimension 128. This lattice is even and self-dual, so that the resulting lattice CFT is indeed holomorphic. Moreover its shortest vectors have length squared 8, so that the lightest winding-momentum states of the CFT have weight 4. 

Next we orbifolded the lattice CFT by the extraspecial group $E(7)$. This is a subgroup of the full automorphism group $BRW(7)$ of the lattice. It has $2^{15}$ elements. It is not Abelian, but it is as close to an Abelian group as possible in the sense that all its commutators commute with all group elements. As with all orbifolds, this leads to untwisted and twisted sectors. In the untwisted sector, the only states of weight less than 4 that survive the orbifold are the vacuum and its Virasoro descendants; all other states are projected out. For the twisted sectors, we established that all states have weight 4 or higher. In total our construction thus has no primary fields of weight less than 4, giving a gap of 4.

One motivation for constructing this CFT is to interpret it in the context of holography as being dual to quantum gravity on $AdS_3$. Because the central charge is relatively small, its dual gravity theory would be in the highly quantum regime. Since the lightest primary field has twice the weight of the energy-momentum tensor, the lightest matter field in the gravity theory would appear at twice the AdS scale. The AdS scale and the matter scale would thus be separated, even though that separation is of course not parametrically large. We note that the CFT is not an extremal CFT and that thus the lightest primary fields do not describe BTZ black holes, which would appear at weight 6. Of course, in using this language we extrapolate from the classical to the highly quantum regime; to properly use such holographic language we would need a large $c$ limit of a family of such CFTs, which we currently do not have. Nonetheless, the conjectured CFT provides the first example of a CFT where the mass scale is strictly bigger than the AdS scale.

In the interest of using the conjectured CFT for applications, let us now gather further information on its structure.

\subsection{Spectrum}

First, let us discuss its spectrum. The lightest states are given by $\vac, L_{-2}\vac, L_{-3}\vac$. The lightest primaries appear at weight 4. On the one hand, primaries in the untwisted sector appear: First, there are 2797 states coming from the weight 4 free boson descendants in the Heisenberg module that are invariant under the action of $E(7)$. 
Second, there are states coming from the vectors of the BW lattice of length squared 8. Their number is given by the number of orbits of such vectors under $E(7)$, since each such orbit corresponds to the state given by the sum of lattice vector states, which remains invariant under $E(7)$. To count the orbits, we use the fact that $BRW(7)$ is the normalizer of $E(7)$, so that it simply permutes $E(7)$ orbits with each other. Moreover $BRW(7)$ acts transitively on the shortest vectors of length squared 8 (see Theorem 11.2 in \cite{MR2159298}). It follows that all such orbits have the same length. Picking any specific orbit, a straightforward calculation gives length 256. There are thus $1260230400 / 256 = 4922775$ weight 4 states coming from lattice vectors.

On the other hand, there are also weight 4 states coming from the twisted sectors associated to elements conjugate to $\Si{11}$ in $E(7)$. (The lightest states in the twisted sectors for the group elements $-\id$ and $\Si{11}\Si{12}$ have weight 6 and 8 respectively.) From (\ref{S11conj}) there are $4^7+2^7-2=16510$ such elements, each of which lives in a conjugacy class of size 2 with respect to $E(7)$, giving 8255 twisted sectors. Each sector has $\sqrt{N}=2^{2^{7-3}}=2^{16}$ ground states, where the finite group $N$ is described at the end of section~\ref{ss:intertw}. This gives a total of $540999680$ twisted states at weight 4. Since the intertwining representation in the $\Si{11}$ sectors is trivial, all these states survive the orbifold projection.

In principle one can continue this analysis also for weight 5 states. Since the partition function is a modular function that is completely determined by the first six terms, knowing the spectrum up to weight 5 is thus enough to fix it completely. In this way one could thus obtain the full spectrum of the CFT. We will leave this for future work.

\subsection{Correlation functions}
Next, let us discuss how to compute correlation functions in the CFT. As usual, this reduces to the problem of determining the 3pt-functions, that is the structure constants of the holomorphic CFT. For the untwisted sector, this is in principle straightforward: Here, it is enough to restrict to the $E(7)$ invariant subsector of the original lattice CFT, whose structure constants are known and are essentially fixed by charge conservation. In particular the correlation functions only involving the lightest states, \ie Virasoro descendants, will be the same as always.

Correlation functions involving twisted sector fields are more difficult. Here the missing ingredients are the fusion rules between the twisted sectors. One way to obtain them would be to compute the $S$-matrices of the twisted sector characters under modular transformations and then apply the Verlinde formula. Unfortunately obtaining exact expressions for those characters and then extracting the $S$-matrices is very hard, since among other things it also requires constructing the action of the orbifold group on the twisted sectors. In principle we could obtain this by the methods described in section~\ref{s:defrep}. However, instead of just computing a single matrix element we would now have to construct the entire representation matrices, which is of course computationally much harder. We therefore leave this for future work. Perhaps less ambitiously, one could investigate the structure of the group multiplication of the conjugacy classes of the twisted sectors, requiring that the identity appear. This would provide at least partial selection rules for the structure constants.

\subsection{Symmetries}
Finally, let us say a few words about the symmetry of the orbifold CFT $V^{orb(E(7)}$. In general, obtaining the automorphisms of an orbifold CFT is a difficult question for various reasons, as we will describe below. 

Let us first focus on $E(7)$. Clearly the untwisted sector $V^{E(7)}$ and the twisted sectors will each have $E(7)$ symmetry separately. Moreover $E(7)$ will permute the twisted sectors. However, in order for the orbifold theory $V^{orb(E(7))}$ to have $E(7)$ symmetry, this permutation action has to be compatible with the fusion rules of the twisted sector, which we did not compute.

The next question is if the orbifold CFT has more symmetry than that. The first observation is that the orbifold CFT has a finite symmetry group since it has no states of weight 1. Next, since the underlying lattice has the much larger automorphism group $BRW(7)$, it is natural to speculate that the orbifold CFT also has at least $BRW(7)$ symmetry. In fact, $BRW(7)$ is the normalizer of $E(7)$, which means that its action is compatible with the orbifold projection onto $E(7)$ invariant states. In the untwisted sector, $BRW(7)$ is thus an automorphism, and it also permutes the twisted modules. However, one would again have to check if it also compatible with the fusion rules.

However, the issue is that $BRW(7)$ is the automorphism group of the lattice, and not of the CFT: the symmetry group of the CFT is rather the lift $\hat{BRW(7)}$. Unlike for $E(7)$, we did not establish that the lifted group of $BRW(7)$ remains the same, and in fact we see no reason why this should be. But if $E(7)$ is not a normal subgroup of $\hat{BRW(7)}$, then the discussion in the above paragraph does not apply.

Finally, there is also the possibility of extra automorphisms, that is symmetries of the orbifold CFT that do not come from automorphisms of the original CFT. The most famous example of this is of course the Monster CFT, where the original Conway group symmetry gets enhanced to the Monster group. To identify such symmetries, a much more detailed analysis of the resulting orbifold CFT would be necessary. We again leave this for future work.

\section*{Acknowledgements}
We thank Klaus Lux and Sven M\"oller  for helpful discussions and comments on the draft. We thank Thomas Gem\"unden and Kristine Van for initial collaboration on this project. We thank an anonymous referee for comments on the draft.
The work of CAK is supported in part by the Simons Foundation
Grant No. 629215. 
The work of CAK and JR is supported by NSF Grant 2111748.

\appendix

\section{Barnes-Wall lattices and their automorphisms}\label{s:BW}
\subsection{The Barnes-Wall lattice}
We will use the construction of \cite{NR02} of the Barnes-Wall lattice. Let $M=\begin{pmatrix}\sqrt 2 & 0\\1&1\end{pmatrix}$ be the generator matrix for the two basis vectors $u_0 = (\sqrt 2,  0 )$ and $u_1= (1,  1)$. The balanced Barnes-Wall lattice is the $\Z[\sqrt 2]$-lattice with generator matrix $\bigotimes^m M$. The Barnes-Wall lattice is then obtained as the integer sublattice of the balanced Barnes-Wall lattice rescaled by ${\sqrt 2}^{-\lceil m/2 \rceil}$.

It will be useful to use binary notation for vector indices. For $i=0,\ldots, 2^m-1$ and $a=0,\ldots,m-1$, let $i[a]$ denote the $a+1$-th bit in the binary notation of $i$, that is $i=(i[m-1]\cdots i[0])$ in binary notation. Let the parity $\pa(i)=0,1$ of $i$ denote the number of zeros of its binary notation of length $m$ modulo 2. Moreover we say $i$ and $j$ are complementary if $i[a] =\neg j[a]$ for all $a$, that is if $i$ is the bitwise complement of $j$.
The basis vectors of the Barnes-Wall lattice $BW(m)$ are then
\be\label{BWbasis}
w_i = {\sqrt 2}^{\pa(i)-\lceil m/2 \rceil} u_{i[m-1]}\otimes \cdots \otimes u_{i[0]}
\ee
Here the first term in the exponent comes from picking the integral sublattice and the second from the overall rescaling. For $m>1$, $BW(m)$ is even. For $m$ odd, it is unimodular, and for $m$ even its Gram matrix has determinant $2^{2^{m-1}}$. Its shortest vectors have length squared $2^{\lfloor m/2\rfloor}$.

From \ref{BWbasis} it is straightforward to compute the Gram matrix and check the above properties. In addition, we will need the following  results on the entries of the Gram matrix:
\begin{lem} \label{lem:Gram}
Let $m\geq 4$.
\begin{enumerate}
\item If $m$ is even, then $(G_m)_{ii}$ is divisible by 4, and $(G_m)_{ij}$ can only be odd if $\pa(i)=\pa(j)=0$.
\item If $m$ is odd, then $(G_m)_{ii}$ is divisible by 4. Moreover we can write $G_m= G^1_m+G^2_m$ where $(G^1_m)_{ij}$ can only be odd if $i$ and $j$ are complementary, and $(G^2_m)_{ij}$ can only be odd if $\pa(i)=\pa(j)=0$.
\end{enumerate}
\end{lem}
\begin{proof}
We have 
\be\label{Gmij}
(G_m)_{ij}=(w_i,w_j)= \sqrt{2}^{\pa(i)+\pa(j)-2\lceil m/2 \rceil + \sum_{l=0}^{m-1}(2i[l]j[l]-i[l]-j[l]+2)}\ .
\ee
It is straightforward to check that the exponent is an even integer.
\begin{enumerate}
\item Let $m=2m'$. It follows
\be
(G_m)_{ii} = 2^{\pa(i)+m'}
\ee
which is divisible by 4 since $m'\geq 2$. For a general entry we can bound the exponent from below by $\pa(i)+\pa(j)$, which establishes the claim.
\item We have
\be
(G_m)_{ii} = 2^{\pa(i)+(m-1)/2}
\ee
which is divisible by 4 since $m\geq 5$. Next consider $(G_m)_{ij}$.  Any even entries of $G_m$ we can distribute at will between $G^1_m$ and $G^2_m$. The only odd entries we can get are if the exponent in (\ref{Gmij}) is 0. We have
\be\label{Gmoddmexp}
m-1 + \sum_{l=0}^{m-1}2i[l]j[l]-i[l]-j[l] \geq -1
\ee
with equality iff $i$ and $j$ are complementary, that is $i =\neg j$, the bitwise complement of $j$. Note that because $m$ is odd we then must have $\pa(i)+\pa(j)=1$, indeed giving an exponent of 0. From these entries we make up $G^1_m$. If (\ref{Gmoddmexp}) is equal to 0, then we can get an exponent of 0 only if $\pa(i)=\pa(j)=0$. From those entries we make up $G^2_m$. Since these are the only two possibilities of having a zero exponent, this establishes the claim.
\end{enumerate}
\end{proof}

\subsection{$E(m)$ and its lifts}\label{ss:appAut}
The automorphism group of the Barnes-Wall lattice is $Aut(BW(m))= BRW(m)$. As discussed in section~\ref{s:BWEm}, it contains an extraspecial group $E(m)=2_+^{2m+1}$ as a subgroup. 
We choose a presentation 
\be
\sigma_1 = \begin{pmatrix}0 & 1\\1&0 \end{pmatrix}\ , \qquad \sigma_2 = \begin{pmatrix}1 & 0\\0&-1 \end{pmatrix}\ .
\ee
Then $E(m)$ is generated by $z=-\id$ and
\be
\Si{ar}= \id_2 \otimes \cdots \id_2 \otimes \sigma_r \otimes\id_2 \cdots \otimes \id_2 \qquad r=1,2\qquad a =1,2,\ldots, m\ .
\ee
They satisfy the defining relations of $E(m)$
\be
[ z, \Si{ar}]=1 \qquad  z^2 =  \Si{ar}^2=1 \qquad [ \Si{ar}, \Si{bs}]=  z^{\delta_{a,b}\delta_{r,s+1\mod 2}}
\ee
Let us now give the matrices in the lattice basis.
The generators act on the lattice vectors (\ref{BWbasis}) as
\bea
w_i z &=& - w_i\\
w_i \Si{a1} &=& (-1)^{\delta_{i[m-a],0}} w_i + \delta_{i[m-a],0} 2^{\pa(i)}w_{i+e_{m-a}}\\
w_i \Si{a2} &=& (-1)^{\delta_{i[m-a],1}} w_i + \delta_{i[m-a],1} 2^{\pa(i)}w_{i-e_{m-a}}
\eea
We note in passing that after a suitable permutation of the basis vectors, the matrices for the generators $\Si{11}$ and $\Si{12}$ are
\be\label{Si11Si12form}
\Si{11}=\begin{pmatrix} -I & 0 & 2I & 0\\ 0 &-I &0 &I\\ 0 & 0&I&0\\ 0 & 0 & 0& I\end{pmatrix} \ ,
\qquad 
\Si{12}= \begin{pmatrix} I & 0 & 0 & 0\\ 0 &I &0 &0\\ I & 0&-I&0\\ 0 & 2I & 0& -I\end{pmatrix}\ ,
\ee
where $I$ and $0$ are block matrices of size $2^{m-2}$. Similar expressions hold for all other $\Si{ar}$ after different permutations of the basis vectors. For our purposes however we of course need explicit expressions for the same choice of basis for all generators.

For the purposes of lifting, we are interested in these generators modulo 2.
Let us define $\Si{ar}=I + A_{ar}$. 
Modulo 2 we then have
\bea
w_i A_{a1} \mod 2 &=& \left\{\begin{array}{ccc} w_{i+e_{m-a}} &: & \pa(i)=0 , \quad i[m-a]=0 \\ 0 &: & \textrm{else} \end{array}\right . \\
w_i A_{a2} \mod 2 &=& \left\{\begin{array}{ccc} w_{i-e_{m-a}} &: & \pa(i)=0 , \quad i[m-a]=1 \\ 0 &: & \textrm{else} \end{array}\right .
\eea
Let $V_0$ and $V_1$ be the subspaces spanned by the vectors $w_i$ of parity $\pa(i)=0$ and 1 respectively. We then have
\be
V_1 \subset \ker A_{ar} \qquad \textrm{and}\qquad \textrm{im}\ A_{ar}\subset V_1\ ,
\ee
so that in particular
\be\label{Anilpotent}
A_{ar}A_{bs}=0 
\ee
for all $a,b,s,r$.

As a reminder, we have $u_{\Sigma}(\alpha) =(-1)^{B_\Sigma(\alpha,\alpha)}$ with
\be
B_\Sigma(\alpha_i,\alpha_j) = \left\{ \begin{array}{cc} \alpha_i H \alpha_j^T - \alpha_i\Sigma H \Sigma^T\alpha_j &: i>j\\ 0 &:i\leq j \end{array} \right. ,
\ee
where $H$ is the `half-Gram matrix' defined as
\be
H_{ij} = \left\{ \begin{array}{cc} G_{ij} &: i>j\\ G_{ij}/2 &: i=j \\ 0 &:i<j \end{array} \right . 
\ee

\begin{lem}\label{lem:ueven}
Let $m\geq4$ be even. Then $u_\Si{ar}=1$ for all generators.
\end{lem}
\begin{proof}
We want to show that
\be
A H + H A^T + A H A^T = 0 \mod 2\ . \label{u1eq}
\ee
From this it then follows that $B_\Si{ar}=0 \mod 2$.
First note that lemma~\ref{lem:Gram} implies for diagonal entries $H_{ii}\mod 2 =0$. For $i\neq j$, note that the image of $A$ has odd parity. But lemma~\ref{lem:Gram} implies that $H_{ij}=0 \mod 2$ unless $\pa(i)=\pa(j)=0$, establishing the claim. 
\end{proof}

\begin{lem}\label{lem:uodd}
Let $m\geq 5$ be odd and $i>j$. Then the following identities hold $\mod 2$:
\begin{enumerate}
\item $(A_{ar}H)_{ij} = (HA^T_{ar})_{ij}$ for $a=2,3,\ldots, m$
\item $(A_{12}H)_{ij}=0$ \label{A12H}
\item $(HA_{11}^T)_{ij}=0$
\end{enumerate}
\end{lem}
\begin{proof}
First, use (2) in Lemma~\ref{lem:Gram} to write $H=H_1+H_2$. Here $H_2$ again vanishes on odd parity subspaces, so we can apply the same logic as in Lemma~\ref{lem:ueven} to argue that the $H_2$ part of all identities works. It is thus enough to show the identities for $H_1$ only. Note that $H_1$ acts as
\be
 w_i H_1 = \delta_{i>\neg i} w_{\neg i}
\ee
Here $\neg i$ indicates the bitwise complement, and $\delta_{i>j}$ is 1 if $i>j$ and 0 otherwise.
\begin{enumerate}
\item Consider $w_iA_{a2}H_1$  and $w_iH_1A_{a2}^T$. Assume $\pa[i]=0$ and $i[m-a]=1$, since otherwise both terms vanish individually. (Note that because $m$ is odd, $\neg$ flips parity.) We
have $w_iA_{a2}H_1= w_{\neg(i-e_{m-a})}$ if $i[m-1]=1$ and 0 otherwise, since the digit $j[m-1]$ determines whether $j>\neg j$. Similarly we have $w_iH_1A_{a2}^T=w_{\neg i + e_{m-a}}= w_{\neg(i-e_{m-a}}$ again if $i[m-1]=1$ and 0 otherwise. It follows that $w_iA_{a2}H_1 = w_iH_1A_{a2}^T$. A similar argument works for $A_{a1}$. 
\end{enumerate}
Note that this argument only works because $A_{ai}$ leaves $j[m-1]$ invariant for $a>1$. For $a=1$ we do indeed find a different outcome:
\begin{enumerate}[resume]
\item Consider $w_iA_{12}H_1$. Assume $i[m-1]=1$, since otherwise the identity is automatically satisfied. But then $w_iA_{12} = w_j$ where $j[m-1]=0$, so that $\neg j > j$ and therefore $w_jH_1=0$.
\item Consider $w_iH_1A_{11}^T$. Assume $i[m-1]=1$, since otherwise $\neg i> i$ and therefore $w_iH_1=0$. But then $w_iH_1=j$ where $j[m-1]=0$, so that $w_j A_{11}^T=0$.
\end{enumerate}
\end{proof}
Note that \eg $w_{01\ldots 10}A_{11} H_1 = w_{0\ldots 01}$, so that indeed $u_{\Si{11}}\neq 1$.

\begin{lem}\label{lem:utriv}
For $m\geq 5$ odd, $r=1,2$ and $a=2,3,\ldots m$, $u_{\Si{ar}}=1$.
\end{lem}
\begin{proof}
We again want to establish (\ref{u1eq}) for $i>j$. The first two terms cancel by (1) in Lemma~\ref{lem:uodd}. For the third term we use (1) again to write $AHA^T=A^2H$, and then note that $A^2=0$ because the image of $A$ has odd parity and is therefore in the kernel of $A$.
\end{proof}

\begin{prop}\label{prop:appcov}
For $m\geq4$, we have $u_{\Si{ar}}(\alpha \Si{bs})=u_{\Si{ar}}(\alpha)$ for all generators. 
\end{prop}
\begin{proof}
In view of Lemmas~\ref{lem:ueven} and \ref{lem:utriv}, this is trivially true if $m$ is even or $a>1$. For $m$ odd and $a=1$ we have
\be
\Si{bs}B_{\Si{11}}\Si{bs}^T= B_{\Si{11}}+ A_{bs}A_{11}H + A_{11}HA_{bs}^T + A_{bs}A_{11}HA_{bs}^T= B_{\Si{11}}+ A_{11}HA_{bs}^T
\ee
by (\ref{Anilpotent}).  Next note that for $m$ odd, $H$ flips the parity of all vectors. This means that $\im A_{11} H \subset V_0 \subset \ker A_{bs}^T$ so that the last term vanishes. The same argument works for $u_{\Si{12}}$.
\end{proof}

\section{Conjugacy classes of $E(m)$ in $BRW(m)$}\label{ss:Emcon}
In this appendix we prove the statement about conjugacy classes of $E(m)$ in $BRW(m)$.
Let us first discuss conjugation in $BRW(2)$:
\begin{lem}\label{lem:m2con}
Under conjugation by elements of $BRW(2)$, we have
\begin{enumerate}
\item $\sigma_3\otimes\sigma_3 \sim  \sigma_2\otimes\sigma_1$ \label{m21}
\item $\sigma_1\otimes \sigma_1 \sim  \id_2\otimes \sigma_1$\label{m22}
\item $\sigma_2\otimes \sigma_1 \sim \id_2\otimes \sigma_1$ \label{m23}
\item $\sigma_2\otimes \sigma_2\sim \sigma_2\otimes \id_2$\label{m24}
\item $\sigma_1\otimes \id_2\sim \sigma_2\otimes \id_2$, $\sigma_1\otimes\sigma_1 \sim \sigma_2\otimes \sigma_2$\label{m25}
\item $\sigma_3\otimes \sigma_2\sim \sigma_3\otimes \id_2$ \label{m26}
\end{enumerate}
\end{lem}
\begin{proof}
Conjugate by the following elements:
\begin{enumerate}
\item $\begin{pmatrix}1&1\\0&1\end{pmatrix}\in GL(2,2)$
\item $\begin{pmatrix}1&1\\1&0\end{pmatrix}\in GL(2,2)$
\item $D(2)$
\item $\begin{pmatrix}1&1\\0&1\end{pmatrix}\in GL(2,2)$
\item $h\otimes h$. (Note that this element has an even number of $h$ and therefore has Dickson invariant 0.)
\item $\begin{pmatrix}1&1\\0&1\end{pmatrix}\in GL(2,2)$
\end{enumerate}
\end{proof}
We can then use this to discuss conjugation for $m>2$:
\begin{lem}\label{lem:mcon}
Take $m\geq 2$. Define $\sigma_0=\id_2$ and $\sigma_3=\sigma_1\sigma_2$.  Then any element $g\in E(m)$,
\be
g = \pm \sigma_{i_1} \otimes \ldots \otimes \sigma_{i_m}\qquad i_j =0,1,2,3\ ,
\ee
is conjugate in $BRW(m)$ to either $\id, -\id, \sigma_2\otimes\id_2 \otimes \ldots \otimes \id_2$ or $\sigma_3\otimes\id_2 \otimes \ldots \otimes \id_2$.
\end{lem}

\begin{proof}
For simplicity assume $m\geq 3$; the case $m=2$ is a straightforward application of lemma~\ref{lem:m2con}.
First note that we can use elements in $GL(m,2)$ to permute the basis vectors in $\F_2$. In terms of $\R^{2^m}$ this operation corresponds to permuting the factors of $g$. The idea is then to use lemma~\ref{lem:m2con} on pairs of factors since the $m=2$ conjugation matrices can be embedded in $BRW(m)$ by tensoring with identity matrices $\id_2$.

Start with an element $g$ of the form above.
If all the factors are $\id_2$, then we are immediately in the case $\pm\id$. Assume therefore that there is at least one non-trivial factor. Then use (\ref{m21}) until there is only one or no $\sigma_3$ left.

If there is no $\sigma_3$, we use (\ref{m22}), (\ref{m23}) and (\ref{m24}) until we are left with a single $\sigma_{1}$ or $\sigma_2$, which we can turn into $\sigma_2$ using (\ref{m25}) if needed. Finally we can conjugate by $\sigma_1$ to switch sign if needed. This establishes that the element is conjugate to $\sigma_2\otimes\id_2 \otimes \ldots \otimes \id_2$.

If there is one $\sigma_3$, we apply the above procedure to the remaining factors until we are left with just one $\sigma_3$ and one $\sigma_2$. We then use (\ref{m26}) to eliminate $\sigma_2$. Finally we can conjugate by $\sigma_1$ to switch sign if needed. The element is thus conjugate to $\sigma_3\otimes\id_2 \otimes \ldots \otimes \id_2$
\end{proof}

\section{Evaluating Schur Averages}\label{app:SchurAve}
\subsection{Schur averaging}
In this section we give an explicit expression for a given entry of the Schur average of an elementary matrix. The central idea is to reduce the problem to a problem in linear algebra over $\F_2$. The quadratic form that appears in the exponent of the Schur average can then be brought to a standard form. In this new basis, the sum factorizes into sums over the elements of at most 2-dimensional vector spaces. These can easily be evaluated, giving a result in closed form.

Take the Schur average of an elementary matrix with entry at $\ct{e}, \ct{f}$:
\be
T^{\ct{e}\ct{f}}_h(e_\ct{c})= \frac1{\sqrt{|N|}} \sum_{\substack{a\in N :\\ \ct{c}+a\hm\pi_2=\ct{e}}} 
(-1)^{a\Delta^h_{1/2}a^T+a  R_{1/2}a^T+ a\hm R_{1/2}(\ct{c}\iota_2)^T+ aR_{1/2}(\ct{f}\iota_2)^T+as_h^T}
e_{\ct{f}+a\pi_2}\ .
\ee
Let us read off the matrix element $\ct{c},\ct{d}$. This means that we only sum over the elements of $N$ which satisfy the linear equations
\be\label{ls}
a\hm \pi_2 = \ct{e}+\ct{c}\ , \qquad a\pi_2 = \ct{f}+\ct{d}\ ,
\ee
which we will write as $a\Pi = (\ct{e}+\ct{c},\ct{f}+\ct{d})$.
To evaluate this sum efficiently, we do the following: first consider the kernel of the linear system (\ref{ls}), $K:=\ker \Pi = \ker \hm\pi_2 \cap \ker \pi_2$. Let $A$ be a basis matrix of this kernel, and $\vec y$ a coordinate vector so that the kernel is parametrized by $\vec y A$. Next let $x$ be a particular solution of (\ref{ls}). If there is no particular solution for theses values of $\ct{c},\ct{d},\ct{e},\ct{f}$, then the matrix entry vanishes.

Otherwise we parametrize $a=x+\vec y A$ and sum over $\vec y$. This gives the following expression for the matrix element 
\be
(T^{\ct{e}\ct{f}}_h)_{\ct{c}\ct{d}}
=  \frac1{\sqrt{|N|}}\sum_{\vec y} (-1)^{\vec y \alpha^h \vec y^T + \beta^h(\ct{c},\ct{d},\ct{e},\ct{f})\vec y^T+ \gamma^h(\ct{c},\ct{d},\ct{e},\ct{f})}
\ee
where
\bea
\alpha^h &=& A M^h_{1/2} A^T\\
\beta^h(\ct{c},\ct{d},\ct{e},\ct{f}) &=& x M^h A^T+ \ct{c} \iota_2 R^T_{1/2} \hm^T A^T +\ct{f}\iota_2 R^T_{1/2} A^T+s_hA^T\\
\gamma^h(\ct{c},\ct{d},\ct{e},\ct{f}) &=& xM^h_{1/2}x^T + x(\hm R_{1/2}(\ct{c}\iota_2)^T+ R_{1/2} (\ct{f}\iota_2)^T+s_h^T)
\eea
with $M^h_{1/2}=\Delta^h_{1/2}+R_{1/2} $ and $M^h = M^h_{1/2}+(M^h_{1/2})^T$. We note that $M^h$ is symmetric and has zeros on the diagonal. It follows that $AM^hA^T=0$, so that $\beta^h$ is independent of the choice of $x$. 

There are $2^{\dim K}$ terms in this sum, so evaluating it by brute force quickly becomes too hard. Instead we want to bring the quadratic function in the exponent to a standard form, and then factorize the sum into products of sums over at most 4 elements. To this, we first find a change of basis $S_h$ that brings the quadratic form $\alpha^h$ to standard form $\tilde \alpha^h$ --- see section~\ref{ss:Qstand} for a description of the algorithm. Note that $\alpha^h$ does not depend on any of the indices $\ct{c},\ct{d},\ct{e},\ct{f}$.

The standard form $\tilde \alpha^h = S_h \alpha^h S_h^T$ is given by
\be\label{alphaStandard}
\vec y \tilde \alpha^h \vec y^T = \sum_{i=1}^r y_{2i-1} y_{2i}+ \sum_{i=2r+1}^{2r+s} 1\cdot y_i^2 + \sum_{i=2r+s+1}^{\dim K} 0\cdot y_i^2\ .
\ee
That is, the corresponding bilinear form matrix is block diagonal with $r$ $2\times 2$ matrices and 1's and 0's on the remaining diagonal. (In principle there can be an additional defect term, but in our case this is either absent or can be included in the last term in (\ref{alphaStandard}).)
This means that after this change of basis, we can factorize the sum. Defining $\tilde\beta^h = \beta^h S^T_h $, for the first $r$ coordinate pairs we have
\be
\sum_{y_{2i-1}, y_i\in \Z_2} (-1)^{y_{2i-1}y_{2i} + \tilde \beta^h_{2i-1}y_{2i-1}+\tilde \beta^h_{2i}y_{2i}} = (-1)^{ \tilde \beta^h_{2i-1}\tilde \beta^h_{2i}}
\ee
The remaining coordinates then factor into sums. Because there is only a linear term (since $y_i^2=y_i$), each sum is either 0 or 2, depending if $\tilde \beta^h_i$ is 0 or 1. More precisely, if we define $\pi_d^h$ to be the projector onto the last $\dim K-2r$ entries and $\tilde\beta^h_d= \tilde\beta^h\pi^h_d$,  then the sum vanishes unless
\be\label{betacondition}
\tilde\beta^h_d = \bar\beta
\ee
where
\be
\bar \beta := (\underbrace{1,\ldots,1}_s,\underbrace{0,\ldots,0}_{\dim K-s-2r})\ .
\ee

Putting everything together we get
\be\label{Thshortcut}
(T^{\ct{e}\ct{f}}_h)_{\ct{c}\ct{d}}
=\left\{ \begin{array}{ccc} 
\frac{2^{\dim K-2r}}{\sqrt{|N|}} (-1)^{\gamma^h +  \sum_{i=1}^r \tilde \beta^h_{2i-1}\tilde \beta^h_{2i}}
&:& \tilde \beta^h_d = \bar \beta
\textrm{\ and\ } (\ct{e}+\ct{c},\ct{f}+\ct{d}) \in \im \Pi\\
0 &:& \textrm{else}
\end{array}\right. 
\ee

\subsection{Finding non-vanishing intertwiners}
Let us now discuss how we can systematically find $\ct{e}, \ct{f}$ so that at least some entry $\ct{c},\ct{d}$ is non-vanishing. First, let $B$ be a basis matrix for $\coim \Pi$. Since $B\Pi$ is a basis matrix for $\im B$, an arbitrary vector in $\im \Pi$ is given by $\vec z B\Pi$, where $\vec z$ is a vector of dimension $\dim \im \Pi$. Our approach is then to find a solution for $\ct{c},\ct{f},\vec z$ that leads to a tuple $\ct{e}, \ct{f},\ct{ c}, \ct{ d}$ satisfying the non-vanishing condition in (\ref{Thshortcut}).

To do this we first parametrize $\ct{e},\ct{d}$ as
\be
(\ct{e},\ct{d})=(\ct{c},\ct{f})+ \vec z B\Pi\ . 
\ee
This means a particular solution satisfying $x\Pi =(\ct{e}+\ct{c},\ct{f}+\ct{d})$ is given by
\be
x = \vec z B \ .
\ee
To see if (\ref{betacondition}) can be satisfied, we need to check if 
\be
(\vec z BM^h + \ct{c} \iota_2 R^T_{1/2} \hm^T  +\ct{f}\iota_2 R^T_{1/2} +s_h)A^TS_h^T\pi^h_d= \bar \beta
\ee
has a solution in $\vec z, \ct{c},\ct{f}$. That is, we need to check if
\be
\bar \beta+s_hA^TS_h^T\pi^h_d
\in \langle B M^h A^T S_h^T\pi^h_d, \iota_2 R^T_{1/2} \hm^T A^T S_h^T\pi^h_d,\iota_2 R^T_{1/2} A^TS_h^T\pi^h_d\rangle\ .
\ee
If so, then $(T^{\ct{e}\ct{f}}_h)$ will have a non-vanishing entry at $\ct c, \ct d$.

\subsection{Multiplying elementary intertwiners}
Let $T^1$ and $T^2$ be two intertwiner matrices for group elements $h_1,h_2$ coming from elementary matrices $(\ct{e}_1,\ct{f}_1), (\ct{e}_2,\ct{f}_2)$. We want to efficiently compute the matrix element of their product,
\be\label{Tproduct}
\sum_{\ct{b}} T^1_{\ct{c}\ct{b}}T^2_{\ct{b}\ct{d}}\ .
\ee
We want to find the set of $\ct{b}$ which gives non-vanishing terms. This set we can again characterize as the solution to a linear system.
First, let $B_1$ be a basis matrix for $\coim \Pi_1$ and $B_2$ a basis matrix for $\coim \Pi_{2}$. We then have
\be
\ct{b}= \ct{f}_1+\vec z_1 B_1\Pi_1\pi_2 = \ct{e}_2 +\vec z_2 B_2\Pi_2 \pi_1\ ,
\ee
where as usual $\pi_1,\pi_2$ are the projectors onto the first and second half of the coordinates respectively, and
\bea
\vec z_1 B_1\Pi_1\pi_1&=&\ct{e}_1+\ct{c}\\
\vec z_2 B_2\Pi_2\pi_2&=&\ct{f}_2+\ct{d}
\eea
We are then looking for solutions $\vec z_1,\vec z_2$ of the system
\bea
\vec z_1 B_1 M^1 A_1^T S_1^T\pi^1_d &=& -(\ct{c} \iota_2 R^T_{1/2} \hm_1^T  -\ct{f_1}\iota_2 R^T_{1/2} -s_1)A_1^TS_1^T\pi^1_d + \bar \beta_1 \\
\vec z_2 (B_2 M^2  + B_2\Pi_2\pi_1\iota_2 R^T_{1/2} \hm_2^T) A_2^T S_2^T\pi^2_d 
&=& -(\ct{e}_2 \iota_2 R^T_{1/2} \hm_2^T -\ct{f}_2\iota_2 R^T_{1/2} -s_2)A_2^TS_2^T\pi^2_d + \bar \beta_2\\
\vec z_1 B_1\Pi_1 \pi_2 - \vec z_2 B_2\Pi_2\pi_1 &=& -\ct{f}_1+ \ct{e}_2\\
\vec z_1 B_1\Pi_1\pi_1&=&\ct{e}_1+\ct{c}\\
\vec z_2 B_2\Pi_2\pi_2&=&\ct{f}_2+\ct{d}
\eea
Let $Z_1$ be the set of all $\vec z_1$ for which there is a solution $(\vec z_1, \vec z_2)$ of the above system. The sum in (\ref{Tproduct}) is then over the set
\be
\ct{b}\in \ct{f}_1 + Z_1B_1\Pi_1\pi_2\ . 
\ee

\subsection{Quadratic forms over characteristic 2}\label{ss:Qstand}
Let $Q$ be a (possibly degenerate) quadratic form over $\F_2$. Let its polar form be $\beta(u,v):=Q(u+v)-Q(u)-Q(v)$. For completeness, let us give an explicit implementation of an algorithm that brings it to standard form. This is a straightforward generalization of the usual standard result in \eg \cite{MR1189139} to the degenerate case.
\begin{enumerate}
\item Find a basis of $\ker \beta$. Note that $Q(u+v)=Q(u)+Q(v)$ if $u\in \ker \beta$.
\item \label{step2} Let $W:=\mathrm{coim}\ \beta$. Note that $\beta$ is non-degenerate on $W$. As long as $\dim W >2$, choose three linearly independent vectors $u_{1,2,3}\in W$.
\item Construct a vector $v\in W$ with $Q(v)=0$:
\begin{itemize}
\item If $Q(u_i)=0$, $v=u_i$.
\item Else if $\beta(u_i,u_j)=0$, $v=u_i+u_j$.
\item Else $v=u_1+u_2+u_3$.
\end{itemize}
\item Since $\beta$ is non-degenerate in $W$, pick $w\in W$ such that $\beta(v,w)=1$. Then define $u=w+Q(w)v$. $Q$ then has the desired form on $\langle v,u \rangle$, namely $Q(av+bu) = ab$.
\item Construct $\langle v,u \rangle^\perp \subset W$ by using the orthogonalization $w \mapsto w - \beta(w,v)u-\beta(w,u)v$.
\item Continue with step~\ref{step2} with $\langle v,u \rangle^\perp \subset W$.
\item Once $\dim \leq 2$, we are at the Witt defect described \eg in \cite{MR1189139}. For our purposes, this never occurs.
\end{enumerate}

\section{Data for the intertwiner representations for $m=5$ and $m=7$}
In the following we give the necessary data to construct the intertwiner matrices $T$ for $m=5$ and $m=7$. For convenience we use sparse array notation: for any vector we give the set of coordinates that have entry 1, all other entries being 0. Perhaps more usefully, all this data can also be found in the supplementary Magma and Mathematica files.

\subsection{Intertwiner representation for $m=5$}\label{ss:intm5}
Let $B$ be the $32\times 32$ basis transformation matrix that brings the bilinear form $C$ to its standard symplectic form described in section~\ref{ss:dr} via $B C B^T$. The group generators in this basis are then given by $B\Sigma B^{-1}$.
We give the row vectors of $B$ in table~\ref{Bm5}.
\begin{table}
\begin{tabular}{cccc}
 1 & \{1\} & 17 & \{32\} \\
 2 & \{1,17\} & 18 & \{8\} \\
 3 & \{1,17,29\} & 19 & \{2,17,29\} \\
 4 & \{1,17,29,31\} & 20 & \{2,17,20,29,31\} \\
 5 & \{1,17,27,31\} & 21 & \{5,17,29,31\} \\
 6 & \{1,17,27,28,31\} & 22 & \{5,17,22,29,31\} \\
 7 & \{1,17,26,28\} & 23 & \{3,17,20,27,31\} \\
 8 & \{1,17,26,28,30\} & 24 & \{3,17,20,23,27\} \\
 9 & \{1,17,25\} & 25 & \{8,16,20,22,23,28,30,31\} \\
 10 & \{1,9,26,27,29\} & 26 & \{4,20,28\} \\
 11 & \{1,9,15,16,23,24,26,28,30\} & 27 & \{4,20,24,28\} \\
 12 & \{4,18,20,24,28\} & 28 & \{13,15,16,23,24,25,27,28,31\} \\
 13 & \{4,12,16,18,26,27,28,30,31\} & 29 & \{6,18,22,24,30\} \\
 14 & \{4,11,12,16,18,25,26,30\} & 30 & \{6,18,21,22,24,30\} \\
 15 & \{6,13,14,16,18,25,26,28\} & 31 & \{7,21,23,24,31\} \\
 16 & \{6,10,13,14,16,30\} & 32 & \{7,19,21,23,24,31\} \\
\end{tabular}
\caption{The 32 row vectors of the basis transformation matrix $B$ for $m=5$}\label{Bm5}
\end{table}
For the $B_h$ we take the following vectors $s_h$:
\bea
s_{12} & =& \{19,21\} \\
s_{22} &=&  \{3,4\}\ \\
s_{51} &=&  \{3,4\} 
\eea
and $s_h = 0$ else.  
For the intertwiner matrices we take the Schur average of the following elementary matrices with entries at $\ct e, \ct f$:
\bea
T_{12} &:& \ct e = \{3,4\}\ , \ct f = 0 \\
T_{22} &:& \ct e = \{10\}\ , \ct f = 0 \\
T_{51} &:& \ct e = \{3,4\}\ , \ct f = 0 \\
\textrm{else} &:& \ct e = 0\ , \ct f= 0
\eea
We then normalize the generators by the following normalization factors:
\be
\vec n =(1, 2^{2}, 1, 2^{3}, 2^{2}, 2^{1}, 2^{1}, 2^{1}, 2^{1}, 2^{1})
\ee
Finally the entries of the matrix (\ref{commMatm5}) are obtained by reading off the $\ct c, \ct d$ entries of the product matrices. Here $\ct c = 0$ and the value of $\ct d$ is given in the following table:
\be
\left(
\begin{array}{cccccccccc}
 \{\} & \{10\} & \{\} & \{10\} & \{\} & \{\} & \{\} & \{\} & \{3,4\} & \{\} \\
 \{10\} & \{\} & \{10\} & \{\} & \{10\} & \{10\} & \{10\} & \{10\} & \{\} & \{10\} \\
 \{\} & \{10\} & \{\} & \{10\} & \{\} & \{\} & \{\} & \{\} & \{3,4\} & \{\} \\
 \{10\} & \{\} & \{10\} & \{\} & \{3,4\} & \{10\} & \{10\} & \{10\} & \{\} & \{10\} \\
 \{\} & \{10\} & \{\} & \{3,4\} & \{\} & \{\} & \{\} & \{\} & \{3,4,13,14\} & \{\} \\
 \{\} & \{10\} & \{\} & \{10\} & \{\} & \{\} & \{\} & \{\} & \{3,4\} & \{\} \\
 \{\} & \{10\} & \{\} & \{10\} & \{\} & \{\} & \{\} & \{\} & \{3,4\} & \{\} \\
 \{\} & \{10\} & \{\} & \{10\} & \{\} & \{\} & \{\} & \{\} & \{3,4\} & \{\} \\
 \{3,4\} & \{\} & \{3,4\} & \{\} & \{3,4,13,14\} & \{3,4\} & \{3,4\} & \{3,4\} & \{\} & \{3,4\} \\
 \{\} & \{10\} & \{\} & \{10\} & \{\} & \{\} & \{\} & \{\} & \{3,4\} & \{\} \\
\end{array}
\right)
\ee

\subsection{Intertwiner representation for $m=7$}\label{ss:intm7}
\begin{table}
\begin{footnotesize}
\begin{tabular}{cc}
 1 & \{1\} \\
 2 & \{1,65\} \\
 3 & \{1,65,113\} \\
 4 & \{1,65,113,125\} \\
 5 & \{1,65,113,125,127\} \\
 6 & \{1,65,113,123,127\} \\
 7 & \{1,65,113,123,124,127\} \\
 8 & \{1,65,113,122,124\} \\
 9 & \{1,65,113,122,124,126\} \\
 10 & \{1,65,113,121\} \\
 11 & \{1,65,105,121\} \\
 12 & \{1,65,105,111,121,127\} \\
 13 & \{1,65,105,111,112,121,127\} \\
 14 & \{1,65,105,110,112,121,126\} \\
 15 & \{1,65,105,109,121\} \\
 16 & \{1,65,101,109\} \\
 17 & \{1,65,101,109,119,127\} \\
 18 & \{1,65,101,109,119,120,127\} \\
 19 & \{1,65,101,109,118,120,126\} \\
 20 & \{1,65,101,109,117\} \\
 21 & \{1,65,99\} \\
 22 & \{1,65,99,116,120,124\} \\
 23 & \{1,65,99,115\} \\
 24 & \{1,65,98\} \\
 25 & \{1,65,98,114\} \\
 26 & \{1,65,97\} \\
 27 & \{1,60,64,65,92,99,122,123,126,127\} \\
 28 & \{1,60,63,65,92,95,99,111,112,119,120,122,124,125,127\} \\
 29 & \{1,60,62,63,64,65,92,94,95,99,101,105,109,111,112,119,120,121\} \\
 30 & \{1,59,60,64,65,92,99,115,121,122,124,126\} \\
 31 & \{1,59,60,64,65,92,99,108,112,115,121,122,126\} \\
 32 & \{1,59,60,64,65,91,92,99,108,112,115,121,122,126\} \\
 33 & \{1,54,65,101,109,114,117\} \\
 34 & \{1,54,56,64,65,88,99,114,117,119,120,127\} \\
 35 & \{1,54,55,56,64,65,88,99,114,115\} \\
 36 & \{1,54,55,56,64,65,88,99,104,112,114,115,120\} \\
 37 & \{1,54,55,56,64,65,87,88,99,104,112,114,115,120\} \\
 38 & \{1,52,65,99,114,115\} \\
 39 & \{1,48,64,65,80,104,105,108,111,112,120,121,124,127\} \\
 40 & \{1,47,48,64,65,80,104,105,108,109,120,121,124\} \\
 41 & \{1,47,48,64,65,80,104,105,107,108,109,120,121,124\} \\
 42 & \{1,47,48,64,65,79,80,104,105,107,108,109,120,121,124\} \\
 43 & \{1,47,48,64,65,79,80,103,104,105,107,108,109,120,121,124\} \\
 44 & \{1,46,65,105,109,121\} \\
 45 & \{1,33,66,67,69,73,81,98,99,101,105,113\} \\
 46 & \{1,33,57,59,62,64,66,67,69,73,81,89,91,94,98,99,105,107,109,115,121\} \\
 47 & \{1,33,57,59,61,66,67,69,73,81,89,91,98,99,101,107,109,115,117,121,122,124,126\} \\
 48 & \{1,33,53,54,55,61,66,67,69,73,81,85,87,98,99,103,105,109,114,115,121\} \\
 49 & \{1,33,53,54,55,61,66,67,69,73,81,85,87,93,98,99,103,105,109,114,115,121\} \\
 50 & \{1,33,49,66,67,69,73,97,98,99,101,105,114,115,117,121\} \\
 51 & \{7,67,69,71,72,79,87,103\} \\
 52 & \{7,39,47,55,67,69,72,99,101,104,107,109,112,115,117,120\} \\
 53 & \{7,39,40,47,55,67,69,99,101,103,107,109,115,117\} \\
 54 & \{7,23,49,51,52,59,67,69,72,79,81,85,88,91,95,97,99,103,107,117,120,127\} \\
 55 & \{7,23,49,51,52,59,67,69,72,79,81,84,85,88,91,95,97,99,103,107,117,120,127\} \\
 56 & \{11,43,47,59,67,73,76,99,103,105,108,109,112,115,121,124\} \\
 57 & \{11,43,44,47,59,67,73,99,103,105,106,107,109,115,121\} \\
 58 & \{11,43,44,47,59,67,73,99,102,103,105,106,107,109,115,121\} \\
 59 & \{13,21,25,29,61,79,86,87,90,91,93,94,95,102,106,126,127\} \\
 60 & \{13,21,25,29,61,78,79,86,87,90,91,93,94,95,126,127\} \\
 61 & \{24,72,84,86,87,88,96,120\} \\
 62 & \{28,76,84,90,91,92,96,124\} \\
 63 & \{30,78,86,90,93,94,96,126\} \\
 64 & \{31,79,87,91,93,95,96,127\} \\
\end{tabular}
\end{footnotesize}
\caption{Row vectors of the basis transformation matrix $B$ for $m=7$\label{Bm71}}
\end{table}

\begin{table}
\begin{footnotesize}
\begin{tabular}{cc}
 65 & \{128\} \\
 66 & \{32\} \\
 67 & \{8\} \\
 68 & \{2,65,113,125\} \\
 69 & \{2,65,68,113,125,127\} \\
 70 & \{5,65,113,123,127\} \\
 71 & \{5,65,70,113,123,124,127\} \\
 72 & \{3,65,68,113,122,124\} \\
 73 & \{3,65,68,71,113,122,124,126\} \\
 74 & \{8,68,70,71,80\} \\
 75 & \{20,68\} \\
 76 & \{17,65,113,121\} \\
 77 & \{17,65,82,113,121\} \\
 78 & \{19\} \\
 79 & \{20,68,82,83,88\} \\
 80 & \{12,68,80,112\} \\
 81 & \{9,65,105,109,121\} \\
 82 & \{9,65,74,105,109,121\} \\
 83 & \{11\} \\
 84 & \{12,68,74,75,80,92,112,124\} \\
 85 & \{14,70,74,80,99,105,109,112,121\} \\
 86 & \{13\} \\
 87 & \{14,70,74,77,80,94,99,105,112,115,121,126\} \\
 88 & \{15,71,75,77,80,111,127\} \\
 89 & \{15,71,75,77,80,95,111\} \\
 90 & \{32,64,80,88,92,94,95,112,120,124,126,127\} \\
 91 & \{6,70\} \\
 92 & \{4,68\} \\
 93 & \{4,66,68\} \\
 94 & \{6,66,70,85,117\} \\
 95 & \{6,38,66,85,98,105,109,117,121\} \\
 96 & \{6,38,66,69,98,105,109,121\} \\
 97 & \{11,67,75,91\} \\
 98 & \{10,66,74\} \\
 99 & \{10,66,74,89,91,121\} \\
 100 & \{10,42,54,55,56,64,66,88,89,91,98,99,104,108,114,115,120,122,126\} \\
 101 & \{10,42,54,55,56,64,66,73,87,88,98,99,104,108,114,115,120,121,122,126\} \\
 102 & \{13,69,73,77,109\} \\
 103 & \{18,66,82,114\} \\
 104 & \{6,22,54,66,69,82,85,88,94,114,117,120,126\} \\
 105 & \{6,22,50,52,66,69,85,88,94,98,113,115,120,122,124\} \\
 106 & \{10,26,50,52,54,66,73,89,92,94,98,113,114,115,117,124,126\} \\
 107 & \{10,18,26,73,81,89,92,94,114,122\} \\
 108 & \{19,67,81,83,87,91,115\} \\
 109 & \{16,79,80,112\} \\
 110 & \{4,66,67,68\} \\
 111 & \{4,66,67,68,72\} \\
 112 & \{4,36,52,66,67,72,98,99,104,108,114,115,120,124\} \\
 113 & \{4,36,52,66,67,72,76,98,99,104,108,114,115,120,124\} \\
 114 & \{16,72,76,79,80,96,112\} \\
 115 & \{58,114,121,122,124,126\} \\
 116 & \{25,73,81,89,91,93,121\} \\
 117 & \{25,73,81,89,90,91,93,121\} \\
 118 & \{13,45,46,47,61,69,73,101,103,105,107,109,117,121\} \\
 119 & \{13,45,46,47,61,69,73,101,103,105,106,107,109,117,121\} \\
 120 & \{21,69,81,85,87,93,117\} \\
 121 &
   \{13,27,29,43,44,47,49,51,52,55,59,61,69,73,76,81,84,85,87,92,93,94,97,105,109,115,117,
   121,124,126\} \\
 122 &
   \{13,21,27,29,43,44,47,49,51,52,55,59,61,73,76,84,86,92,94,97,102,105,109,115,121,124,1
   26\} \\
 123 & \{19,49,51,52,55,59,67,97,99,102,103,107\} \\
 124 & \{19,49,51,52,55,59,67,97,99,100,103,107\} \\
 125 & \{41,73,97,105,106,107,109,121\} \\
 126 & \{37,69,97,101,102,103,109,117\} \\
 127 & \{35,67,97,99,100,103,107,115\} \\
 128 & \{34,66,97,98,100,102,106,114\} \\
\end{tabular}
\end{footnotesize}
\caption{Row vectors of the basis transformation matrix $B$ for $m=7$, continued\label{Bm72}}
\end{table}

The row vectors of the $128\times128$ basis transformation matrix $B$ are given in tables~\ref{Bm71} and \ref{Bm72}.
For the $B_h$ we take the following vectors $s_h$:
\bea
s_{12} & =& \{68,70,100\} \\
s_{22} &=&  \{4,5\}\ \\
s_{32} &=&  \{100\} \\
s_{71} &=&\{ 67,70,91,95,120\}
\eea
and $s_h = 0$ else.  
For our elementary matrices we choose $\ct e =\ct f=0$ except for $T_{22}$, for which we choose
\be
T_{22}\ : \ \ct e =\{4,5\}\ , \qquad \ct f =\{47\}\ .
\ee
We normalize the generators by the following normalization factors:
\be
\vec n =(1, 2^{8}, 1, 2^{10}, 2^{1}, 2^{4}, 2^{3}, 2^{2}, 2^{2}, 2^{4}, 2^{3}, 2^{4}, 2^{6}, 1)
\ee
For the relation matrix we always compute the entry of the product with $\ct c=0$ and $\ct d =0$. The only exception is for $T_{22}$, for which we take $\ct c =0$ and the $\ct d$ vectors are given by
\begin{multline}
(\{4,5,28\},\{4,5,46\},\{4,5,28\},\{\},\{4,5,28\},\{4,5,28\},\{4,5,28\},\{4,5,28\},\\
\{4,5,28\},\{4,5,28\},\{4,5,28,63\},
\{4,5,28\},\{4,5,28,64\},\{4,5,28\})
\end{multline}
Here the $k$-th entry is the $\ct d$ vector for the product of $T_{22}$ with the $k$-th intertwiner matrix.




\bibliographystyle{alpha}
 \bibliography{./refmain}

\end{document}